\documentclass[letterpaper]{article} % DO NOT CHANGE THIS
\usepackage{aaai25}  % DO NOT CHANGE THIS
\usepackage{times}  % DO NOT CHANGE THIS
\usepackage{helvet}  % DO NOT CHANGE THIS
\usepackage{courier}  % DO NOT CHANGE THIS
\usepackage[hyphens]{url}  % DO NOT CHANGE THIS
\usepackage{graphicx} % DO NOT CHANGE THIS
\usepackage{natbib}  % DO NOT CHANGE THIS AND DO NOT ADD ANY OPTIONS TO IT
\usepackage{caption} % DO NOT CHANGE THIS AND DO NOT ADD ANY OPTIONS TO IT
\usepackage{algorithm}
\usepackage{algorithmic}

\usepackage{mathbbol}

\usepackage{tikz} % nice language for creating drawings and diagrams
\usepackage{amsmath,amsfonts,amssymb,dsfont,bm,amsthm}
\usepackage{csquotes}
\usepackage{mathbbol}
\usepackage{subcaption}
\usepackage{thmtools,thm-restate}
\usepackage{MnSymbol}
\usepackage{bm}
\usepackage{accents}

\usepackage{ciphod}

\newtheorem{definition}{Definition}

\newtheorem{theorem}{Theorem}

\usepackage{newfloat}
\usepackage{listings}
\DeclareCaptionStyle{ruled}{labelfont=normalfont,labelsep=colon,strut=off} % DO NOT CHANGE THIS
\floatstyle{ruled}
\newfloat{listing}{tb}{lst}{}
\floatname{listing}{Listing}
\title{Identifying Macro Conditional Independencies and Macro Total Effects in Summary Causal Graphs with Latent Confounding
\\[1ex] \large including appendix}

\author{
    Simon Ferreira\textsuperscript{\rm 1, 2},
    Charles K. Assaad\textsuperscript{\rm 1}
}
\affiliations{
    \textsuperscript{\rm 1}Sorbonne Université, INSERM, Institut Pierre Louis d’Epidémiologie et de Santé Publique, F75012, Paris, France\\
    \textsuperscript{\rm 2}ENS de Lyon, F69342, Lyon, France\\
    simon.ferreira@sorbonne-universite.fr, charles.assaad@inserm.fr
}

\usepackage{bibentry}
\begin{document}

\maketitle

\begin{abstract}
Understanding causal relations in dynamic systems is essential 
in epidemiology. 
While causal inference methods have been extensively studied, they often rely on fully specified causal graphs, which may not always be available 
in complex dynamic systems. Partially specified causal graphs, and in particular summary causal graphs (SCGs), provide a simplified representation of causal relations between time series when working spacio-temporal data, omitting temporal information and focusing on causal structures between clusters of of temporal variables. 
Unlike fully specified causal graphs, SCGs can contain cycles, which complicate their analysis and interpretation.
In addition, their cluster-based nature introduces new challenges concerning the types of queries of interest: macro queries, which involve relationships between clusters represented as vertices in the graph, and micro queries, which pertain to relationships between variables that are not directly visible through the vertices of the graph.
In this paper, we first clearly distinguish between macro conditional independencies and micro conditional independencies and between macro total effects and micro total effects. Then, we demonstrate the soundness and completeness of the \emph{d-separation} to identify macro conditional independencies in SCGs. Furthermore, we establish that the \emph{do-calculus} is sound and complete for identifying macro total effects in SCGs. 
Finally, we give a graphical characterization for the non-identifiability of macro total effects in SCGs.
%Conversely, we also show through various examples that these results do not hold when considering micro conditional independencies and micro total effects.
\end{abstract}

\begin{figure*}[!ht]
	\centering
	\begin{subfigure}{.29\textwidth}
		\centering
		\begin{tikzpicture}[{black, circle, draw, inner sep=0}]
			\tikzset{nodes={draw,rounded corners},minimum height=0.6cm,minimum width=0.6cm}	
			\tikzset{latent/.append style={white, fill=black}}
			
			% \node[latent] (L) at (1.75,1.2) {$L$};
			\node[fill=red!30] (X) at (0,0) {$X$} ;
			\node[fill=blue!30] (Y) at (2.4,0) {$Y$};
			\node (Z) at (1.2,0) {$W$};

			\begin{scope}[transform canvas={yshift=-.25em}]
				\draw [->,>=latex] (X) -- (Z);
			\end{scope}
			\begin{scope}[transform canvas={yshift=.25em}]
				\draw [<-,>=latex] (X) -- (Z);
			\end{scope}			
			
			\begin{scope}[transform canvas={yshift=-.25em}]
				\draw [->,>=latex] (Z) -- (Y);
			\end{scope}
			\begin{scope}[transform canvas={yshift=.25em}]
				\draw [<-,>=latex] (Z) -- (Y);
			\end{scope}

			% \draw[->,>=latex, dashed] (L) to [out=0,in=90, looseness=1] (Y);
			% \draw[->,>=latex, dashed] (L) to [out=180,in=90, looseness=1] (X);
			\draw[<->,>=latex, dashed] (X) to [out=90,in=90, looseness=1] (Y);
			
			\draw[->,>=latex] (X) to [out=165,in=120, looseness=2] (X);
			\draw[->,>=latex] (Y) to [out=15,in=60, looseness=2] (Y);
			% \draw[->,>=latex] (L) to [out=165,in=120, looseness=2] (L);
			\draw[->,>=latex] (Z) to [out=15,in=60, looseness=2] (Z);
			
			\draw[<->,>=latex, dashed] (X) to [out=-155,in=-110, looseness=2] (X);
			\draw[<->,>=latex, dashed] (Y) to [out=-25,in=-70, looseness=2] (Y);
			\draw[<->,>=latex, dashed] (Z) to [out=-25,in=-70, looseness=2] (Z);			
		\end{tikzpicture}
		\caption{SCG.}
		\label{fig:SCG_0}
	\end{subfigure}
	\hfill 
	\begin{subfigure}{.33\textwidth}
		\centering
		\begin{tikzpicture}[{black, circle, draw, inner sep=0}]
			\tikzset{nodes={draw,rounded corners},minimum height=0.7cm,minimum width=0.6cm, font=\scriptsize}
			\tikzset{latent/.append style={white, fill=black}}
			
			\node  (Z) at (1.2,1.2) {$W_t$};
			\node (Z-1) at (0,1.2) {$W_{t-1}$};
			\node  (Z-2) at (-1.2,1.2) {$W_{t-2}$};
			\node[fill=blue!30] (Y) at (1.2,0) {$Y_t$};
			\node[fill=blue!30] (Y-1) at (0,0) {$Y_{t-1}$};
			\node[fill=blue!30]  (Y-2) at (-1.2,0) {$Y_{t-2}$};
			\node[fill=red!30] (X) at (1.2,2.4) {$X_t$};
			\node[fill=red!30] (X-1) at (0,2.4) {$X_{t-1}$};
			\node[fill=red!30]  (X-2) at (-1.2,2.4) {$X_{t-2}$};
			%%% self-loop
			%%% others 
			\draw[->,>=latex] (X-2) to (Z-2);
			\draw[->,>=latex] (X-1) to (Z-1);
			\draw[->,>=latex] (X) to (Z);
			\draw[->,>=latex] (Z-2) -- (Y-2);
			\draw[->,>=latex] (Z-1) -- (Y-1);
			\draw[->,>=latex] (Z) -- (Y);
			
			\draw[->,>=latex] (X-2) to (Z-1);
			\draw[->,>=latex] (X-1) to (Z);
			\draw[->,>=latex] (Z-2) -- (X-1);
			\draw[->,>=latex] (Z-1) -- (X);
			\draw[->,>=latex] (Z-2) -- (Y-1);
			\draw[->,>=latex] (Z-1) -- (Y);
			\draw[->,>=latex] (Y-2) to (Z-1);
			\draw[->,>=latex] (Y-1) to (Z);

			\draw[->,>=latex] (X-2) -- (X-1);
			\draw[->,>=latex] (X-1) -- (X);
			\draw[->,>=latex] (Y-2) -- (Y-1);
			\draw[->,>=latex] (Y-1) -- (Y);
			\draw[->,>=latex] (Z-2) -- (Z-1);
			\draw[->,>=latex] (Z-1) -- (Z);
			
			\draw[<->,>=latex, dashed] (X-2) to [out=180,in=180, looseness=1] (Y-2);
			\draw[<->,>=latex, dashed] (X-1) to [out=30,in=-30, looseness=0.5] (Y-1);
			\draw[<->,>=latex, dashed] (X) to [out=0,in=0, looseness=1] (Y);
			\draw[<->,>=latex, dashed] (X-2) to (Y-1);
			\draw[<->,>=latex, dashed] (X-1) to (Y);
			\draw[<->,>=latex, dashed] (X-2) to [out=80,in=100, looseness=1] (X-1);
			\draw[<->,>=latex, dashed] (X-1) to [out=80,in=100, looseness=1] (X);
			\draw[<->,>=latex, dashed] (Y-2) to [out=-80,in=-100, looseness=1] (Y-1);
			\draw[<->,>=latex, dashed] (Y-1) to [out=-80,in=-100, looseness=1] (Y);

			%\coordinate[left of=V-2] (d1);
			\draw [dashed,>=latex] (X-2) to[left] (-2,2.4);
			\draw [dashed,>=latex] (Z-2) to[left] (-2,1.2);
			\draw [dashed,>=latex] (Y-2) to[left] (-2,0);
			\draw [dashed,>=latex] (X) to[right] (2,2.4);
			\draw [dashed,>=latex] (Z) to[right] (2,1.2);
			\draw [dashed,>=latex] (Y) to[right] (2,0);
		\end{tikzpicture}		\caption{\centering Compatible FT-ADMG 1.}
		\label{fig:example2_FT-ADMG1}
	\end{subfigure}
	\hfill 
	\begin{subfigure}{.33\textwidth}
		\centering
		\begin{tikzpicture}[{black, circle, draw, inner sep=0}]
			\tikzset{nodes={draw,rounded corners},minimum height=0.7cm,minimum width=0.6cm, font=\scriptsize}
			\tikzset{latent/.append style={white, fill=black}}
			
			\node  (Z) at (1.2,1.2) {$W_t$};
			\node (Z-1) at (0,1.2) {$W_{t-1}$};
			\node  (Z-2) at (-1.2,1.2) {$W_{t-2}$};
			\node[fill=blue!30] (Y) at (1.2,0) {$Y_t$};
			\node[fill=blue!30] (Y-1) at (0,0) {$Y_{t-1}$};
			\node[fill=blue!30]  (Y-2) at (-1.2,0) {$Y_{t-2}$};
			\node[fill=red!30] (X) at (1.2,2.4) {$X_t$};
			\node[fill=red!30] (X-1) at (0,2.4) {$X_{t-1}$};
			\node[fill=red!30]  (X-2) at (-1.2,2.4) {$X_{t-2}$};
			%%% self-loop
			%%% others 
			\draw[->,>=latex] (Z-2) to (X-2);
			\draw[->,>=latex] (Z-1) to (X-1);
			\draw[->,>=latex] (Z) to (X);
			\draw[->,>=latex] (Y-2) -- (Z-2);
			\draw[->,>=latex] (Y-1) -- (Z-1);
			\draw[->,>=latex] (Y) -- (Z);
			
			\draw[->,>=latex] (X-2) to (Z-1);
			\draw[->,>=latex] (X-1) to (Z);
			\draw[->,>=latex] (Z-2) -- (X-1);
			\draw[->,>=latex] (Z-1) -- (X);
			\draw[->,>=latex] (Z-2) -- (Y-1);
			\draw[->,>=latex] (Z-1) -- (Y);
			\draw[->,>=latex] (Y-2) to (Z-1);
			\draw[->,>=latex] (Y-1) to (Z);

			\draw[->,>=latex] (X-2) -- (X-1);
			\draw[->,>=latex] (X-1) -- (X);
			\draw[->,>=latex] (Y-2) -- (Y-1);
			\draw[->,>=latex] (Y-1) -- (Y);
			\draw[->,>=latex] (Z-2) -- (Z-1);
			\draw[->,>=latex] (Z-1) -- (Z);
			
			\draw[<->,>=latex, dashed] (X-2) to [out=180,in=180, looseness=1] (Y-2);
			\draw[<->,>=latex, dashed] (X-1) to [out=30,in=-30, looseness=0.5] (Y-1);
			\draw[<->,>=latex, dashed] (X) to [out=0,in=0, looseness=1] (Y);
			\draw[<->,>=latex, dashed] (X-2) to (Y-1);
			\draw[<->,>=latex, dashed] (X-1) to (Y);
			\draw[<->,>=latex, dashed] (X-2) to [out=80,in=100, looseness=1] (X-1);
			\draw[<->,>=latex, dashed] (X-1) to [out=80,in=100, looseness=1] (X);
			\draw[<->,>=latex, dashed] (Y-2) to [out=-80,in=-100, looseness=1] (Y-1);
			\draw[<->,>=latex, dashed] (Y-1) to [out=-80,in=-100, looseness=1] (Y);

			%\coordinate[left of=V-2] (d1);
			\draw [dashed,>=latex] (X-2) to[left] (-2,2.4);
			\draw [dashed,>=latex] (Z-2) to[left] (-2,1.2);
			\draw [dashed,>=latex] (Y-2) to[left] (-2,0);
			\draw [dashed,>=latex] (X) to[right] (2,2.4);
			\draw [dashed,>=latex] (Z) to[right] (2,1.2);
			\draw [dashed,>=latex] (Y) to[right] (2,0);
		\end{tikzpicture}		
		\caption{Compatible FT-ADMG 2.}
		\label{fig:example2_FT-ADMG2}
	\end{subfigure}
	\caption{An SCG with two compatible FT-ADMGs. Each pair of red and blue vertices represents the macro query we are interested in.}
	\label{fig:SCG_FTADMG}
\end{figure*}
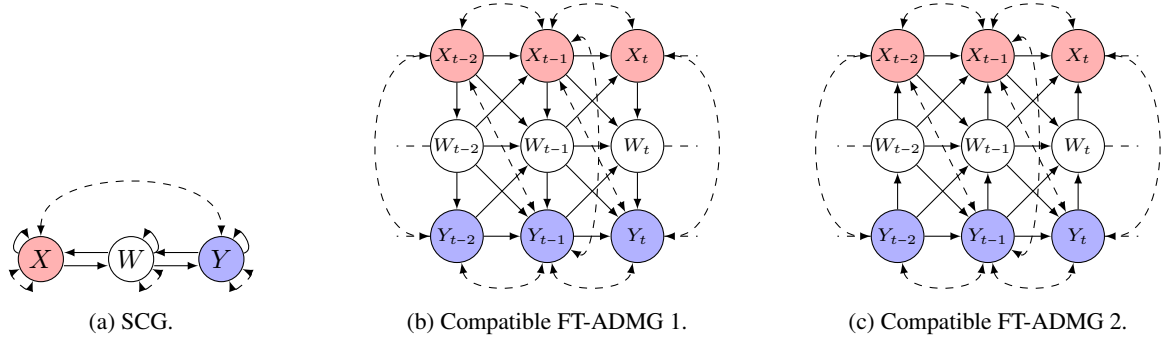

\section{Introduction}

In observational studies, causal graphs are vital for decision-making regarding interventions and can be useful for detecting independencies between variables. Several standard tools have been developed for fully specified causal graphs, typically represented as directed acyclic graphs (DAGs) or acyclic directed mixed graphs (ADMGs). For example, d-separation \citep{Pearl_1988} enables the detection of conditional independencies directly from the graph. Similarly, the do-calculus~\citep{Pearl_1995} allows for the derivation of total effect estimands from observational data. These results were also extended to directed mixed graphs (DMGs) \citep{Richardson_1997, Forre_2017, Forre_2018, Forre_2020}.
However, constructing a fully specified causal graph is challenging because it requires knowledge of the causal relations among all pairs of observed variables. This knowledge is often unavailable, particularly in complex, high-dimensional settings, thus limiting the applicability of causal inference theory and tools.
Therefore, recently there has been more interest in partially specified causal graphs~\citep{Perkovic_2020,Anand_2023,Ferreira_2024,Assaad_2024}.
%,especially in partially understood domains, such as in epidemiology.
An important type of partially specified graphs is the cluster-DMG which provides a coarser representation of causal relations through vertices that represent a cluster of variables (and can contain cycles). 
The partial specification of cluster-DMG arises from the fact that each vertex in an cluster-DMG represents a cluster of variables. 
A well-known special case of cluster-DMGs is the cluster-ADMG, which imposes an acyclicity assumption on the graph~\citep{Anand_2023}. Another notable special case of cluster-DMGs in dynamic systems is the summary causal graph (SCG), where each cluster represents one time series in spacio-temporal data or repeated measurements through time of the same variable in a longitudinal study, and all temporal information is omitted from the graph~\citep{Assaad_2022}.

In SCGs (or any cluster-DMG), there are two primary types of queries of interest. The first, called a micro query, involves considering one or a few specific variables (but not all) represented by a vertex in an SCG. The second, called a macro query, involves considering the whole cluster of variables represented by a vertex in an SCG. 
This paper focuses on two types of macro queries: macro conditional independencies and macro total effects. Specifically, we show that the d-separation criterion and the do-calculus respectively identify macro conditional independencies and macro total effects in SCGs. By addressing these macro queries, we aim to provide practical solutions for inferring causal effects in real-world applications where causal relations are only partially understood.
In particular, the concept of macro queries is crucial for public health strategies and epidemiological modeling. For example, the macro total effect can be used to understand the impact of the COVID-19 pandemic on the flu epidemic and vice versa \citep{Boelle_2020, Ferguson_2003}, shedding light on the broader impacts of interventions and natural immunity dynamics over extended periods.
Note that we assume having access to spacio-temporal data with full time series, i.e., where the coverage from the beginning of generative process is ensured. However, this does not significantly limit the applicability of our results as epidemiologists, for example, often rely on this assumption when analyzing the impact of the COVID-19 pandemic on concurrent epidemics~\cite{Feng2021}. While this assumption is not made explicit in such studies, it remains essential for ensuring the validity of such findings. Furthermore, presenting our results along with a clear articulation of this assumption establishes a solid foundation for future work to address this limitation, which is often overlooked in practical applications.

The remainder of the paper is organized as follows: Firstly, we formally presents SCGs as well as micro and macro queries.
% Section~\ref{sec:disting} shows how SCGs are related to other known objects in causal inference.
Secondly, we show that d-separation is sound and complete for macro conditional independencies in SCGs.
Thirdly, we show that the do-calculus is sound and complete for macro total effects in SCGs and present a graphical characterization for the non-identifiability of these effects.
%Section~\ref{sec:micro} briefly discusses how the results of this paper are related to micro queries.
After that, we discuss related works. %works shows the limitation of this paper and the links with existing articles.
Finally, we conclude the paper while showing its limitations.

\section{Summary Causal Graphs and Macro Queries}
\label{sec:SCGs}

We suppose that in a dynamic system, observations are generated from a discrete-time dynamic structural causal model
(DTDSCM), an extension of structural causal models \citep{Pearl_2000} to dynamic systems. In the following we use the convention that uppercase letters represent variables and lowercase letters represent their values, moreover letters in blackboard bold represent sets.

\begin{definition}[Discrete-time dynamic structural causal model (DTDSCM)]
	\label{def:DTDSCM}
	A discrete-time dynamic structural causal model is a tuple $\mathcal{M}=(\mathbb{L}, \mathbb{V}, \mathbb{F},  \proba{\mathbb{l}})$,
	where $\mathbb{L} = \bigcup \{\mathbb{L}^{v^i_t} \mid i \in [1,d], t \in [t_0,t_{max}]\}$ is a set of exogenous variables, which cannot be observed but affect the rest of the model.
	$\mathbb{V} = \bigcup \{\mathbb{V}^i \mid i \in [1,d]\}$ such that $\forall i \in [1,d]$, $\mathbb{V}^{i} = \{V^{i}_{t} \mid t \in [t_0,t_{max}]\}$, is a set of endogenous variables, which are observed and every $V^i_t \in \mathbb{V}$ is functionally dependent on some subset of $\mathbb{L}^{v^i_t} \cup \mathbb{V}_{\leq t} \backslash \{V^i_t\}$ where $\mathbb{V}_{\leq t} = \{V^j_{t'} \mid j \in [1,d], t'\leq t\}$.
	$\mathbb{F}$ is a set of functions such that for all $V^i_t \in \mathbb{V}$, $f^{v^i_t}$ is a mapping from $\mathbb{L}^{v^i_t}$ and a subset of $\mathbb{V}_{\leq t} \backslash \{V^i_t\}$ to $V^i_t$.
	$\proba{\mathbb{l}}$ is a joint probability distribution over $\mathbb{L}$.
\end{definition}

We suppose that the direct causal relations $\mathbb{F}$ in a DTDSCM can be qualitatively represented by a full-time acyclic directed mixed graph (FT-ADMG), a special case of ADMGs~\citep{Richardson_2003}, where bidirected dashed arrows represent hidden confounding.

\begin{definition}[Full-Time Acyclic Directed Mixed Graph]
	\label{def:FT-ADMG}
	Consider a DTDSCM $\mathcal{M}$. The \emph{full-time acyclic directed mixed graph (FT-ADMG)} $\mathcal{G} = (\mathbb{V}, \mathbb{E})$ induced by $\mathcal{M}$ is defined in the following way:
	\begin{equation*}
		\begin{aligned}
			\mathbb{E}^{1} :=& \{X_{t-\gamma}\rightarrow Y_{t} \mid \forall Y_{t} \in \mathbb{V},~X_{t-\gamma} \in \mathbb{X} \subseteq \mathbb{V}_{\leq t} \backslash \{Y_t\} \\
            &\st Y_t := f^{y_t}(\mathbb{X}, \mathbb{L}^{y_t}) \text{ in } \mathcal{M}\},\\
            \mathbb{E}^{2} :=& \{X_{t-\gamma}\longdashleftrightarrow Y_{t} \mid \forall X_{t-\gamma}, Y_{t} \in \mathbb{V}\\
            &\st \notindep{\mathbb{L}^{x_{t-\gamma}}}{\mathbb{L}^{y_t}\}}{\probasymbol},
		\end{aligned}
	\end{equation*}
	where $\mathbb{E}=\mathbb{E}^{1}\cup \mathbb{E}^{2}$.
\end{definition}

In the remainder, we will use the notations of parents, children, ancestors and descendants for FT-ADMGs. Consider an FT-ADMG $\mathcal{G}=(\mathbb{V},\mathbb{E})$ for every vertex $Y_t \in \mathbb{V}$ we note:
$\parents{Y_t}{\mathcal{G}} = \{X_{t'} \in \mathbb{V} \mid X_{t'} \rightarrow Y_t\}$, $\children{Y_t}{\mathcal{G}} = \{X_{t'} \in \mathbb{V} \mid Y_t \rightarrow X_{t'}\}$, $\ancestors{Y_t}{\mathcal{G}} = \bigcup_{n\in \mathbb{N}}P_n$ where $P_0 = \{Y_t\}$ and $P_{k+1} = \bigcup_{X_{t'} \in P_k} \parents{X_{t'}}{\mathcal{G}}$, $\descendants{Y_t}{\mathcal{G}} = \bigcup_{n\in \mathbb{N}}C_n$ where $C_0 = \{Y_t\}$ and $C_{k+1} = \bigcup_{X_{t'} \in C_k} \children{X_{t'}}{\mathcal{G}}$.
% \begin{itemize}
% 	\item $\parents{Y_t}{\mathcal{G}} = \{X_{t'} \in \mathbb{V} \mid X_{t'} \rightarrow Y_t\}$,\\
% 	\item $\children{Y_t}{\mathcal{G}} = \{X_{t'} \in \mathbb{V} \mid Y_t \rightarrow X_{t'}\}$,\\
% 	\item $\ancestors{Y_t}{\mathcal{G}} = \bigcup_{n\in \mathbb{N}}P_n$ where $P_0 = \{Y_t\}$ and $P_{k+1} = \bigcup_{X_{t'} \in P_k} \parents{X_{t'}}{\mathcal{G}}$, and\\
% 	\item $\descendants{Y_t}{\mathcal{G}} = \bigcup_{n\in \mathbb{N}}C_n$ where $C_0 = \{Y_t\}$ and $C_{k+1} = \bigcup_{X_{t'} \in C_k} \children{X_{t'}}{\mathcal{G}}$.
% \end{itemize}

In many domains, such as epidemiology, it is generally challenging for practitioners to validate, analyze, or even provide an FT-ADMG due to the difficulty in determining the temporal lag between a cause and its effect. However, practitioners can usually provide summary causal graphs, which are compact versions of FT-ADMG. These summary causal graphs represent the causal relations between the exposure, the disease and the different factors without specifying the temporal lags of these relations.

\begin{definition}[Summary Causal Graph with possible latent confounding]
	\label{def:SCG}
	Consider an FT-ADMG $\mathcal{G} = (\mathbb{V}, \mathbb{E})$. The \emph{summary causal graph (SCG)} $\mathcal{G}^{s} = (\mathbb{S}, \mathbb{E}^{{s}})$ compatible with $\mathcal{G}$ is defined in the following way:
	\begin{equation*}
		\begin{aligned}
			% &\mathbb{S} &:=& \{Y &\mid& \forall Y_t \in \mathbb{V}\},&\\
			\mathbb{S} :=& \{V^i = (V^i_{t_0},\cdots,V^i_{t_{max}}) \mid \forall i \in [1, d]\},\\
           \mathbb{E}^{s1} :=& \{X\rightarrow Y \mid \forall X,Y \in \mathbb{S},~\exists t'\leq t\in [t_0,t_{max}] \\
        &\st X_{t'}\rightarrow Y_{t}\in\mathbb{E}\},\\
            \mathbb{E}^{s2} :=& \{X\longdashleftrightarrow Y \mid \forall X,Y \in \mathbb{S},~\exists t',t \in [t_0,t_{max}] \\ 
&\st X_{t'}\longdashleftrightarrow Y_{t}\in\mathbb{E}\}.        
\end{aligned}\\
\end{equation*}

where $\mathbb{E}^{s} = \mathbb{E}^{s1}\cup \mathbb{E}^{s2}.$
 
\end{definition}

The transformation from an FT-ADMG to an SCG can be made explicit using the natural transformation from $\mathbb{V}$ to $\mathbb{S}$, $\tau_0$ : 
%\begin{align*}    
$(v^1_{t_0},\cdots,v^1_{t_{max}},\cdots,v^d_{t_0},\cdots,v^d_{t_{max}}) \mapsto ((v^1_{t_0},\cdots,v^1_{t_{max}}),\cdots,(v^d_{t_0},\cdots,v^d_{t_{max}})).$
%\end{align*}
The abstraction of SCGs entails that, even though there is exactly one SCG compatible with a given FT-ADMG, there are in general several FT-ADMGs compatible with a given SCG. For example, we give in Figure~\ref{fig:SCG_FTADMG} an SCG with two of its compatible FT-ADMGs.
While SCGs, like cluster-DMGs, can contain cycles, the cycles in an SCG arise from the partial specificity of the graph.
Notice that, like cluster-DMGs, SCGs may have directed cycles and in particular two directed edges oriented in opposite directions, \ie, if in the FT-ADMG we have $X_{t'}\rightarrow Y_{t}$ and $Y_{t''}\rightarrow X_{t}$ then in the SCG we have $X\rightarrow Y$ and $Y\rightarrow X$ which we often write $X\rightleftarrows Y$.
However, unlike in cluster-DMGs, the cycles in SCGs necessarily arise from the partial specificity of the graph.
Moreover, if $\mathcal{G}^{s} = (\mathbb{S}, \mathbb{E}^{{s}})$ and $\mathcal{G} = (\mathbb{V}, \mathbb{E})$ are compatible, we abuse the notation $\forall V^i \in \mathbb{S}$,~$\mathbb{V}^i=\{V^i_{t} \mid t \in [t_0,t_{max}]\}$ given in Definition~\ref{def:DTDSCM} by writing $\forall V^i = Y \in \mathbb{S}$,~$\mathbb{V}^{Y}=\mathbb{V}^i$ and $\forall \mathbb{Y}\subseteq \mathbb{S}$,~$\mathbb{V}^{\mathbb{Y}}=\bigcup\limits_{Y \in \mathbb{Y}}\mathbb{V}^Y$. 

% In the remainder, as for FT-ADMGs, we will use the notations of parents, children, ancestors and descendants for SCGs. Consider an SCG $\mathcal{G}^{s} = (\mathbb{S}, \mathbb{E}^{{s}})$ for every vertex $Y \in \mathbb{S}$ we note:\\
% \begin{itemize}
% 	\item $\parents{Y}{\mathcal{G}^s} = \{X \in \mathbb{S} \mid X \rightarrow Y\}$,\\
% 	\item $\children{Y}{\mathcal{G}^s} = \{X \in \mathbb{S} \mid Y \rightarrow X\}$,\\
% 	\item $\ancestors{Y}{\mathcal{G}^s} = \bigcup_{n\in \mathbb{N}}P_n$ where $P_0 = \{Y\}$ and $P_{k+1} = \bigcup_{X \in P_k} \parents{X}{\mathcal{G}^s}$, and\\
% 	\item $\descendants{Y}{\mathcal{G}^s} = \bigcup_{n\in \mathbb{N}}C_n$ where $C_0 = \{Y\}$ and $C_{k+1} = \bigcup_{X \in C_k} \children{X}{\mathcal{G}^s}$.
% \end{itemize}
In the remainder, we will adopt the terminology of parents, children, ancestors, and descendants for SCGs, which are defined in a manner consistent with those used for FT-ADMGs, \ie, the sets of parents, children, ancestors, and descendants of a vertex $Y$ in an SCG $\mathcal{G}^s$ are respectively given by $\parents{Y}{\mathcal{G}^s}$, $\children{Y}{\mathcal{G}^s}$, $\ancestors{Y}{\mathcal{G}^s}$, and $\descendants{Y}{\mathcal{G}^s}$. Furthermore, a strongly connected component of $Y$ in $\mathcal{G}^s$ denoted as $\scc{V}{\mathcal{G}^s}$ is the set $\ancestors{Y}{\mathcal{G}^s} \cap \descendants{Y}{\mathcal{G}^s}$.

In this paper, we focus on two distinct queries in SCGs: one that pertains to finding conditional independencies and another that infers total effects. Each of these queries can be categorized into two types: macro query and micro query. 
Below, we define micro conditional independency and macro conditional independency.

\begin{definition}[Micro conditional independency]
	\label{def:micro-indep}
	A micro independence in the FT-ADMG $\mathcal{G} = (\mathbb{V}, \mathbb{E})$ is a conditional independence between two disjoint sets of temporal variables $\mathbb{X}, \mathbb{Y} \subseteq \mathbb{V}$ conditioned on any subset of temporal variables $\mathbb{W} \subseteq \mathbb{V}\backslash (\mathbb{X}, \mathbb{Y})$.
\end{definition}

\begin{definition}[Macro conditional independency]
	\label{def:macro-indep}
	A macro independence in the FT-ADMG $\mathcal{G} = (\mathbb{V}, \mathbb{E})$ compatible with the SCG $\mathcal{G}^{s} = (\mathbb{S}, \mathbb{E}^{{s}})$ is a conditional independence between two sets of temporal variables of the form $\mathbb{V}^X$ and $\mathbb{V}^Y$ for $X, Y \in \mathbb{S}$ conditioned on a subset of temporal variables of the form $\mathbb{V}^\mathbb{W}$ for $\mathbb{W} \subseteq \mathbb{S}\backslash\{X,Y\}$.
\end{definition}

For example, in Figure~\ref{fig:SCG_FTADMG}, we are interested to know if there is a macro conditional independence between $X$ and $Y$ given $W$ in the SCG which means to know if $\{X_{t} \mid t \in [t_0, t_{max}]\}$ is independent of $\{Y_t \mid t \in [t_0, t_{max}]\}$ given $\{W_t \mid t \in [t_0, t_{max}]\}$. In contrast, a micro conditional independence query would be of the form: "is $X_t$ independent of $Y_{t-1}$ given $W_{t-1}$?".

Note that there exists a fundamental tool called d-separation~\citep{Pearl_1988} that serves for deriving conditional independence properties from causal graphs. 
In the following the d-separation of two subsets $\mathbb{X}$ and $\mathbb{Y}$ conditionally on a third one $\mathbb{W}$ in a graph $\mathcal{G}$ is denoted as $\dsepc{\mathbb{X}}{\mathbb{Y}}{\mathbb{W}}{\mathcal{G}}$.

Now, we define micro total effect and macro total effect.
\begin{definition}[Micro total effect]
	\label{def:micro-total_effects}
	A micro total effects in the FT-ADMG $\mathcal{G} = (\mathbb{V}, \mathbb{E})$ is a total effect from a set of temporal variables $\mathbb{X}$ to another set of temporal variables $\mathbb{Y}$ where $\mathbb{X}, \mathbb{Y} \subseteq \mathbb{V}$, denoted as $\probac{\mathbb{Y} = \mathbb{y}}{\interv{\mathbb{X} = \mathbb{x}}}$.
\end{definition}

\begin{definition}[Macro total effect]
	\label{def:macro-total_effects}
	A macro total effect in the FT-ADMG $\mathcal{G} = (\mathbb{V}, \mathbb{E})$ compatible with the SCG $\mathcal{G}^{s} = (\mathbb{S}, \mathbb{E}^{{s}})$ is a total effect from $\mathbb{V}^X$ to $\mathbb{V}^Y$ where $X, Y \in \mathbb{S}$, denoted as $\probac{\mathbb{V}^Y = \mathbb{v}^y}{\interv{\mathbb{V}^X = \mathbb{v}^x}}$.
\end{definition}

By a slight abuse of notation, in the remainder we will denote macro total effects as $\probac{\mathbb{v}^y}{\interv{\mathbb{v}^x}}$.
Notice that macro queries are a special case of micro queries.

Similarly, we will also define the micro and macro interventions in the FT-ADMG respectively as $\mathbb{I}^m = \{\interv{V^i_t = v^i_t} \mid i \in [1,d], t \in [t_0,t_{max}]\}$ and $\mathbb{I}^M = \{\interv{\mathbb{V}^i = \mathbb{v}^i} \mid i \in [1,d]\}$.
Consider also the complete set of interventions in the SCG $\mathbb{I}^s = \{\interv{Y = y} \mid Y \in \mathbb{S}\}$ which corresponds to the macro interventions from the point of view of the SCG.

% \begin{definition}[Micro Interventions]
	%     \label{def:micro-interventions}
	%     The set of micro interventions in the FT-ADMG is the complete set of possible interventions $\mathbb{I}^m=\{\interv{V^i_t=v^i_t}~\mid~1\leq i\leq d,~t_0 \leq t \leq t_{max},~ v^i_t\in Im(f^{v^i_t})\}$.
	% \end{definition}
% Where $Im$ stands for the image of a function.
% \begin{definition}[Macro Interventions]
	%     \label{def:macro-interventions}
	%     The set of macro interventions in the FT-ADMG is $\mathbb{I}^M=\{\interv{V^i_{t_0}=v^i_{t_0},\cdots,V^i_{t_{max}}=v^i_{t_{max}}}~\mid~1\leq i\leq d,~ v^i_t\in Im(f^{v^i_t})\}$.
	% \end{definition}
% Similarly, we will also talk about macro total effects.
% Consider also the complete set of interventions in the SCG $\mathbb{I}^s = \{\interv{Y=(y_{t_0},\cdots,y_{t_{max}}} \mid Y\in \mathbb{S},~ y_t\in Im(f^{y_t})\}$.
% and the natural mapping from $\mathbb{I}^M$ to $\mathbb{I}^s$, $\omega_0: \interv{V^i_{t_0}=v^i_{t_0},\cdots,V^i_{t_{max}}=v^i_{t_{max}}} \mapsto \interv{V^i=(v^i_{t_0},\cdots,v^i_{t_{max}})}$.

\textbf{Remarks}
\begin{itemize}
    \item 
Previous works considering micro queries~\citep{Ferreira_2024,Assaad_2024} usually assume stationarity, because it permits the division of a single multivariate time series into multiple instances, thereby facilitating estimation when each time-point represents a single observation. However, when focusing solely on macro queries, this assumption loses its relevance because stationarity offers no distinct advantage. Instead, it is assumed that there are multiple observations available for each temporal variable, i.e., spacio-temporal data, because estimating the macro total effect from just a single observation per variable is infeasible. 
\item This work addresses macro conditional independencies and macro total effects. However, it is worth noting that the results presented for macro total effects, particularly in Section "Idendification of Macro Total Effects", are also applicable to macro controlled direct effects~\citep{Pearl_2000}.
\end{itemize}

\section{Identification of Macro Conditional Independencies}
\label{sec:indep}
In this section, we show that d-separation is sound and complete in SCGs for macro queries.
Firstly, since an FT-ADMG is an ADMG, the standard definition of d-separation \citep{Pearl_1998} is directly applicable for FT-ADMG.
This definition was introduced for ADMGs and was extended to DMGs in \citep{Forre_2018} and to cluster-ADMGs in \citep{Anand_2023}.
But it turned out that it is readily extendable to SCGs, which are a special case of cluster-DMGs.

\begin{definition}[d-separation in SCGs]
	\label{def:d-separation}
	In an SCG $\mathcal{G}^s=(\mathbb{S},\mathbb{E}^s)$, a path $\pi=\langle V^{p_1},\cdots,V^{p_n}\rangle$ is said to be blocked by a set of vertices $\mathbb{W}\subseteq\mathbb{S}$ if:
	\begin{enumerate}
		\item $\exists 1 < i < n \st V^{p_{i-1}}\starbarstar V^{p_i}\rightarrow V^{p_{i+1}}$ or $V^{p_{i-1}}\leftarrow V^{p_i}\starbarstar V^{p_{i+1}}$ and $V^{p_i}\in\mathbb{W}$, or
		\item $\exists 1 < i < n \st V^{p_{i-1}}\stararrow V^{p_i}\arrowstar V^{p_{i+1}}$ and $\descendants{V^{p_i}}{\mathcal{G}^s}\cap\mathbb{W}=\emptyset$.
	\end{enumerate}
	where $\stararrow$ represents $\rightarrow$ or $\longdashleftrightarrow$, $\arrowstar$ represents $\leftarrow$ or $\longdashleftrightarrow$, and $\starbarstar$ represents any of the three arrow type $\rightarrow$, $\leftarrow$ or $\longdashleftrightarrow$.
	A path which is not blocked is said to be active. A set $\mathbb{W}\subseteq \mathbb{S}$ is said to d-separate two sets of variables $\mathbb{X},~\mathbb{Y}\subseteq \mathbb{S}$ if it blocks every path from a variable of $\mathbb{X}$ to a variable of $\mathbb{Y}$.
\end{definition}

Before providing the main results of this section, we define the notion of walks in graphs and show how it is related to the notion of paths. This notion is important in SCGs since they contain cycles and the relation between paths and walks will be useful in the proof of the soundness of d-separation.

A walk is a sequence of consecutively adjacent vertices and contrary to paths, walks may contain repeating vertices. The definition of blocked path given in Definition~\ref{def:d-separation} is readily extendable to walks.

\begin{restatable}{property}{propertyone}{(Active Walks)}
	\label{prop:active_walks}
	Let $\mathcal{G}=(\mathbb{V},\mathbb{E})$ be an FT-ADMG, $\mathbb{W}\subseteq\mathbb{V}$ and $\tilde{\pi}=\langle V^1,\cdots,V^n\rangle$ be a walk.
	If $\tilde{\pi}$ is $\mathbb{W}$-active then there exist a path $\pi$ from $V^1$ to $V^n$ which is $\mathbb{W}$-active.
\end{restatable}

\begin{proof}
    In Appendix.
\end{proof}

The following theorem shows that the d-separation as presented in Definition~\ref{def:d-separation} is sound in SCGs.

\begin{theorem}[Soundness of d-separation in SCGs]
	\label{theorem:soundness_d-sep}
	Let $\mathcal{G}^s=(\mathbb{S},\mathbb{E}^s)$ be an SCG and $\mathbb{X},\mathbb{Y},\mathbb{W}\subseteq \mathbb{S}$. If $\mathbb{X}$ and $\mathbb{Y}$ are d-separated by $\mathbb{W}$ in $\mathcal{G}^s$ then, in any compatible FT-ADMG $\mathcal{G}=(\mathbb{V},\mathbb{E})$, $\mathbb{V}^{\mathbb{X}}$ and $\mathbb{V}^{\mathbb{Y}}$ are d-separated by $\mathbb{V}^{\mathbb{W}}$.
\end{theorem}

\begin{proof}
	Suppose $\mathbb{X}$ and $\mathbb{Y}$ are d-separated by $\mathbb{W}$ in $\mathcal{G}^s$ and there exists a compatible FT-ADMG $\mathcal{G}=(\mathbb{V},\mathbb{E})$ and a path $\pi=\langle V^{p_1}_{t^{p_1}},\cdots,V^{p_n}_{t^{p_n}}\rangle$ from $V^{p_1}_{t^{p_1}} \in \mathbb{V}^{\mathbb{X}}$ to $V^{p_n}_{t^{p_n}} \in \mathbb{V}^{\mathbb{Y}}$ which is not blocked by $\mathbb{V}^{\mathbb{W}}$.
	The walk $\tilde{\pi}^s=\langle V^{p_1},\cdots,V^{p_n}\rangle$ is in $\mathcal{G}^s$ with $V^{p_1}\in\mathbb{X}$ and $V^{p_n}\in\mathbb{Y}$.
	Since $\mathbb{X}$ and $\mathbb{Y}$ are d-separated by $\mathbb{W}$ and using the contraposition of Property~\ref{prop:active_walks}, $\mathbb{W}$ blocks $\tilde{\pi}^s$ and thus there exists $1<i<n$ such that $\langle V^{p_{i-1}},V^{p_{i}},V^{p_{i+1}} \rangle$ is $\mathbb{W}$-blocked.

	If $V^{p_{i-1}}\starbarstar V^{p_i}\rightarrow V^{p_{i+1}}$ (or symmetrically $V^{p_{i-1}}\leftarrow V^{p_i}\starbarstar V^{p_{i+1}}$) and $V^{p_i}\in\mathbb{W}$ then $V^{p_{i-1}}_{t^{p_{i-1}}} \starbarstar V^{p_i}_{t^{p_{i}}} \rightarrow V^{p_{i+1}}_{t^{p_{i+1}}}$ and $\forall t\in[t_0,t_{max}],~V^{p_i}_t \in \mathbb{V}^{\mathbb{W}}$ and in particular $V^{p_i}_{t^{p_i}} \in \mathbb{V}^{\mathbb{W}}$.
	Therefore, $\pi$ is blocked by $\mathbb{V}^{\mathbb{W}}$.

	Otherwise, $V^{p_{i-1}} \stararrow V^{p_i} \arrowstar V^{p_{i+1}}$ and $\descendants{V^{p_i}}{\mathcal{G}^s} \cap \mathbb{W} = \emptyset$ so $V^{p_{i-1}}_{t^{p_{i-1}}} \stararrow V^{p_i}_{t^{p_{i}}} \arrowstar V^{p_{i+1}}_{t^{p_{i+1}}}$ and $\descendants{V^{p_i}_{t^{p_i}}}{\mathcal{G}} \cap \mathbb{V}^{\mathbb{W}} \subseteq \mathbb{V}^{Desc(V^{p_i})} \cap \mathbb{V}^{\mathbb{W}} = \{\mathbb{V}^j \mid V^j \in \descendants{V^{p_i}}{\mathcal{G}^s} \cap \mathbb{W}\} = \{\mathbb{V}^j \mid V^j \in \emptyset\} = \emptyset$.
	Therefore, $\pi$ is blocked by $\mathbb{V}^{\mathbb{W}}$ which contradicts the initial assumption.
	In conclusion, the d-separation criterion in SCGs is sound.
\end{proof}

The following theorem shows that the d-separation as presented in Definition~\ref{def:d-separation} is complete in SCGs.

\begin{theorem}[Completeness of d-separation in SCGs]
	\label{theorem:completeness_d-sep}
	Let $\mathcal{G}^s=(\mathbb{S},\mathbb{E}^s)$ be an SCG and $\mathbb{X},\mathbb{Y},\mathbb{W}\subseteq \mathbb{S}$. If $\mathbb{X}$ and $\mathbb{Y}$ are not d-separated by $\mathbb{W}$ in $\mathcal{G}^s$ then, there exists a compatible FT-ADMG $\mathcal{G}=(\mathbb{V},\mathbb{E})$ such that $\mathbb{V}^{\mathbb{X}}$ and $\mathbb{V}^{\mathbb{Y}}$ are not d-separated by $\mathbb{V}^{\mathbb{W}}$.
\end{theorem}

\begin{proof}
	Suppose $\mathbb{X}$ and $\mathbb{Y}$ are not d-separated by $\mathbb{W}$ in $\mathcal{G}^s$.
	% If $\mathbb{X}$ and $\mathbb{Y}$ are adjacent in $\mathcal{G}^s$ then in every compatible FT-ADMG, there exists a pair of vertices in $\mathbb{V}^{\mathbb{X}}\times\mathbb{V}^{\mathbb{Y}}$ which are adjacent and thus $\mathbb{V}^{\mathbb{X}}$ and $\mathbb{V}^{\mathbb{Y}}$ are not d-separated by $\mathbb{V}^{\mathbb{W}}$.
	% Now, suppose that $\mathbb{X}$ and $\mathbb{Y}$ are not adjacent in $\mathcal{G}^s$.
	Since $\mathbb{X}$ and $\mathbb{Y}$ are not d-separated by $\mathbb{W}$ in $\mathcal{G}^s$ there exists an $\mathbb{W}$-active path $\pi^s=\langle V^{p_1},\cdots,V^{p_n}\rangle$ with $V^{p_1}\in\mathbb{X}$ and $V^{p_n}\in\mathbb{Y}$.
	From the SCG $\mathcal{G}^s=(\mathbb{S},\mathbb{E}^s)$ one can build a compatible FT-ADMG $\mathcal{G}=(\mathbb{V},\mathbb{E})$ in the following way:
	\begin{equation*}
		\begin{aligned}
			\mathbb{V} :=& \mathbb{V}^{\mathbb{S}}&\\
			\mathbb{E}^{1}_{\pi^s} :=& \{X_{t} \rightarrow Y_{t} \mid \forall t\in[t_0,t_{max}], \forall X \rightarrow Y \in \mathbb{E}^s &\\
			&\st \langle X \rightarrow Y \rangle \subseteq \pi^s \text{ or }\langle Y \leftarrow X \rangle \subseteq \pi^s\}&\\
			\mathbb{E}^{1}_{\bar{\pi}^s} :=& \{X_{t} \rightarrow Y_{t+1} \mid \forall t\in[t_0,t_{max}[, \forall X \rightarrow Y \in \mathbb{E}^s &\\
			&\st \langle X \rightarrow Y \rangle \not\subseteq \pi^s \text{ and }\langle Y \leftarrow X \rangle \not\subseteq \pi^s\}&\\
			\mathbb{E}^{2}_{\pi^s} :=& \{X_{t} \longdashleftrightarrow Y_{t} \mid \forall t\in[t_0,t_{max}], \forall X \longdashleftrightarrow Y \in \mathbb{E}^s &\\
			&\st \langle X \longdashleftrightarrow Y \rangle \subseteq \pi^s\text{ or }\langle Y \longdashleftrightarrow X \rangle \subseteq \pi^s\}&\\
			\mathbb{E}^{2}_{\bar{\pi}^s} :=& \{X_{t} \longdashleftrightarrow Y_{t+1} \mid \forall t\in[t_0,t_{max}[, \forall X \longdashleftrightarrow Y \in \mathbb{E}^s &\\
			&\st \langle X \longdashleftrightarrow Y \rangle \not\subseteq \pi^s\text{ and }\langle Y \longdashleftrightarrow X \rangle \not\subseteq \pi^s\}&\\
			\mathbb{E} :=& \mathbb{E}^{1}_{\pi^s} \cup \mathbb{E}^{1}_{\bar{\pi}^s} \cup \mathbb{E}^{2}_{\pi^s} \cup \mathbb{E}^{2}_{\bar{\pi}^s}
		\end{aligned}
	\end{equation*}
	Notice that $\mathcal{G}$ is in indeed acyclic, contains no edges going back in time and is compatible with the SCG $\mathcal{G}^s$.
	Moreover, $\mathcal{G}$ contains the path $\pi=\langle V^{p_1}_{t_0},\cdots,V^{p_n}_{t_0}\rangle$ which is necessarily $\mathbb{V}^{\mathbb{W}}$-active since $\pi^s$ is $\mathbb{W}$-active.
	In conclusion, the d-separation criterion in SCGs is complete.
\end{proof}

The findings of this section establish that identifying a d-separation in SCGs guarantees a macro-level d-separation in all compatible FT-ADMGs. This is crucial because, as shown in \citealt{Pearl_2000}, a d-separation in an ADMG implies a conditional independence in the underlying probability distribution. By extending the applicability of d-separation to SCGs, this result allows researchers to infer macro-level conditional independencies even when dealing with partially specified graphs that omit specific temporal information. This is particularly valuable in constraint-based causal discovery when the interest is to uncover the structure of the SCG without uncovering an FT-ADMG.

%A bidirected dashed edge  on a single vertex (loop) in an SCG does not affect the results of this section. Therefore, modelers can choose whether to include these edges in this context or not. However, these edges can be more valuable for micro conditional independencies.

\section{Identification of Macro Total Effects}
\label{sec:total_effect}

\begin{figure}[t!]
	\centering
	\begin{subfigure}{.09\textwidth}
		\centering
		\begin{tikzpicture}[{black, circle, draw, inner sep=0}]
			\tikzset{nodes={draw,rounded corners},minimum height=0.6cm,minimum width=0.6cm}	
			\tikzset{anomalous/.append style={fill=easyorange}}
			\tikzset{rc/.append style={fill=easyorange}}
			
			\node (W) at (0,1.2) {$W$} ;
			\node[fill=red!30] (X) at (0,0) {$X$} ;
			\node[fill=blue!30] (Y) at (1.2,0) {$Y$};
			
			\draw [->,>=latex] (X) -- (Y);
			\begin{scope}[transform canvas={xshift=-.25em}]
				\draw [->,>=latex] (X) -- (W);
			\end{scope}
			\begin{scope}[transform canvas={xshift=.25em}]
				\draw [<-,>=latex] (X) -- (W);
			\end{scope}
			
			\draw[->,>=latex] (W) to [out=180,in=135, looseness=2] (W);
			\draw[->,>=latex] (X) to [out=180,in=135, looseness=2] (X);
			\draw[->,>=latex] (Y) to [out=15,in=60, looseness=2] (Y);
			
			\draw[<->,>=latex, dashed] (X) to [out=-155,in=-110, looseness=2] (X);
			\draw[<->,>=latex, dashed] (Y) to [out=-25,in=-70, looseness=2] (Y);
			\draw[<->,>=latex, dashed] (W) to [out=-15,in=-60, looseness=2] (W);	
		\end{tikzpicture}
		\caption{\centering}
		\label{fig:identifiable:a}
	\end{subfigure}
	\hfill  
	\begin{subfigure}{.15\textwidth}
		\centering
		\begin{tikzpicture}[{black, circle, draw, inner sep=0}]
			\tikzset{nodes={draw,rounded corners},minimum height=0.6cm,minimum width=0.6cm}	
			\tikzset{latent/.append style={white, fill=black}}
			
			% \node[latent] (L) at (1.75,1.2) {$L$};
			\node[fill=red!30] (X) at (0,0) {$X$} ;
			\node[fill=blue!30] (Y) at (2.4,0) {$Y$};
			\node (Z) at (1.2,0) {$W$};

			\draw [->,>=latex,] (X) -- (Z);

			\draw[->,>=latex] (Z) -- (Y);

			\draw[<->,>=latex, dashed] (X) to [out=90,in=90, looseness=1] (Y);
			
			\draw[->,>=latex] (X) to [out=165,in=120, looseness=2] (X);
			\draw[->,>=latex] (Y) to [out=15,in=60, looseness=2] (Y);
			% \draw[->,>=latex] (L) to [out=165,in=120, looseness=2] (L);
			\draw[->,>=latex] (Z) to [out=15,in=60, looseness=2] (Z);

			\draw[<->,>=latex, dashed] (X) to [out=-155,in=-110, looseness=2] (X);
			\draw[<->,>=latex, dashed] (Y) to [out=-25,in=-70, looseness=2] (Y);
			\draw[<->,>=latex, dashed] (Z) to [out=-25,in=-70, looseness=2] (Z);			
			
		\end{tikzpicture}
		\caption{}
		\label{fig:identifiable:b}
	\end{subfigure}
	\hfill 
	\begin{subfigure}{.16\textwidth}
		\centering
		\begin{tikzpicture}[{black, circle, draw, inner sep=0}]
			\tikzset{nodes={draw,rounded corners},minimum height=0.6cm,minimum width=0.6cm}	
			\tikzset{latent/.append style={white, fill=black}}
			
			% \node[latent] (L) at (1.75,1.2) {$L$};
			\node[fill=red!30] (X) at (0,0) {$X$} ;
			\node (Z) at (2.4,0) {$Z$};
			\node (W) at (1.2,0) {$W$};
			\node[fill=blue!30] (Y) at (1.8,-1.2) {$Y$};

			\draw [->,>=latex,] (X) -- (Y);
			\draw [->,>=latex,] (X) -- (W);

			\draw[->,>=latex] (W) -- (Z);
			\draw[->,>=latex] (W) -- (Y);

			\draw[<->,>=latex, dashed] (X) to [out=90,in=90, looseness=1] (Z);
			\draw[<->,>=latex, dashed] (W) to [out=-15,in=90, looseness=1] (Y);
			
			\draw[->,>=latex] (X) to [out=165,in=120, looseness=2] (X);
			\draw[->,>=latex] (Y) to [out=15,in=60, looseness=2] (Y);
			% \draw[->,>=latex] (L) to [out=165,in=120, looseness=2] (L);
			\draw[->,>=latex] (W) to [out=15,in=60, looseness=2] (W);
			\draw[->,>=latex] (Z) to [out=15,in=60, looseness=2] (Z);
			
			\draw[<->,>=latex, dashed] (X) to [out=-155,in=-110, looseness=2] (X);
			\draw[<->,>=latex, dashed] (Y) to [out=-25,in=-70, looseness=2] (Y);
			\draw[<->,>=latex, dashed] (W) to [out=-155,in=-110, looseness=2] (W);
			\draw[<->,>=latex, dashed] (Z) to [out=-25,in=-70, looseness=2] (Z);			
		\end{tikzpicture}
		\caption{}
		\label{fig:identifiable:c}
	\end{subfigure}
	\caption{SCGs with identifiable macro total effects. Each pair of red and blue vertices represents the total effect we are interested in.}
	\label{fig:identifiable}
\end{figure}
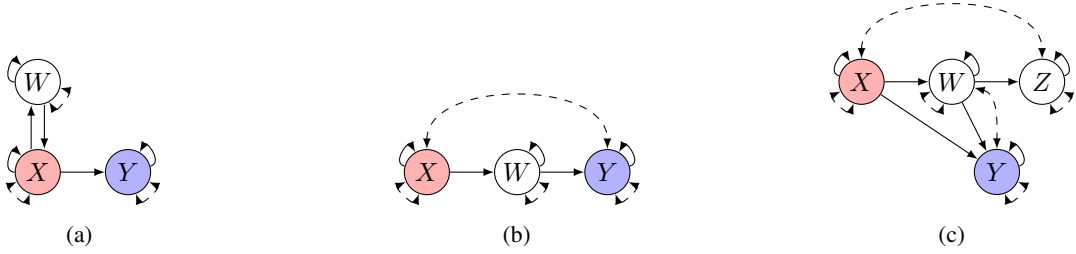

The do-calculus initially introduced in \citealt{Pearl_1995} is an important tool of causal inference that consists of three rules. It allows to express, whenever it is possible, queries under interventions, \ie, that contains a $\interv{\cdot}$ operator, as queries that can be computed from positive observational\footnote{As in \citealt{Pearl_1995}, we assume a strict positive distribution but this assumption can be relaxed to weaker positivity conditions tailored to the specific query expressed solely in terms of observational distributions  as shown in ~\citealt{Hwang_2024}.} distribution, \ie, that does not contain a $\interv{\cdot}$ operator. In such cases, it is said that the query containing the $\interv{\cdot}$ is \textit{identifiable}.
The do-calculus was initially introduced for ADMGs so it is not easily extendable to cyclic graphs.
But it turned out that it is readily extendable to SCGs (containing cycles).

Firstly, we define the notion of mutilated graphs~\citep{Pearl_2000}. Consider a causal graph $\mathcal{G}=(\mathbb{V},\mathbb{E})$ and $\mathbb{A},\mathbb{B}\subseteq\mathbb{V}$, a mutilated graph denoted by $\mathcal{G}_{\bar{\mathbb{A}}\underaccent{\bar}{\mathbb{B}}}$ is the graph obtained by removing all edges coming in $\mathbb{A}$ and all edges coming out of $\mathbb{B}$.
In the following, we introduce a property of compatibility between mutilated graphs that will be useful for proving the soundness and the completeness of the do-calculus in SCGs.

\begin{restatable}{property}{propertytwo}{(Compatibility of Mutilated Graphs)}
	\label{prop:mutilated_graphs}
	Let $\mathcal{G}=(\mathbb{V},\mathbb{E})$ be an FT-ADMG, $\mathcal{G}^s=(\mathbb{S},\mathbb{E}^s)$ its compatible SCG and $\mathbb{A},\mathbb{B}\subseteq\mathbb{S}$.
	The mutilated graph $\mathcal{G}^ s_{\bar{\mathbb{A}}\underaccent{\bar}{\mathbb{B}}}$ is an SCG compatible with the mutilated FT-ADMG $\mathcal{G}_{\bar{\mathbb{V}^{\mathbb{A}}}\underaccent{\bar}{\mathbb{V}^{\mathbb{B}}}}$.
\end{restatable}

\begin{proof}
    In Appendix.
\end{proof}

Using the notion of mutilated graphs and d-separation we show that the do-calculus is applicable to SCGs.

\begin{theorem}[Soundness of the do-calculus]
	\label{theorem:soundness_do-calculus}
	Let $\mathcal{G}^s=(\mathbb{S},\mathbb{E}^s)$ be an SCG and $\mathbb{X},\mathbb{Y},\mathbb{U},\mathbb{W}\subseteq\mathbb{S}$ be disjoint subsets of vertices.
	The three following rules of the do-calculus are sound.
	\begin{equation*}
		\begin{aligned}
			\textbf{Rule 1:}& \probac{\mathbb{v}^{\mathbb{y}}}{\interv{\mathbb{v}^{\mathbb{u}}},\mathbb{v}^{\mathbb{x}},\mathbb{v}^{\mathbb{w}}} = \probac{\mathbb{v}^{\mathbb{y}}}{\interv{\mathbb{v}^{\mathbb{u}}},\mathbb{v}^{\mathbb{w}}}\\
			&\text{if } \dsepc{\mathbb{Y}}{\mathbb{X}}{\mathbb{U}, \mathbb{W}}{\mathcal{G}^s_{\bar{\mathbb{U}}}}\\
			\textbf{Rule 2:}& \probac{\mathbb{v}^{\mathbb{y}}}{\interv{\mathbb{v}^{\mathbb{u}}},\interv{\mathbb{v}^{\mathbb{x}}},\mathbb{v}^{\mathbb{w}}} = \probac{\mathbb{v}^{\mathbb{y}}}{\interv{\mathbb{v}^{\mathbb{u}}},\mathbb{v}^{\mathbb{x}},\mathbb{v}^{\mathbb{w}}}\\
			&\text{if } \dsepc{\mathbb{Y}}{\mathbb{X}}{\mathbb{U}, \mathbb{W}}{\mathcal{G}^s_{\bar{\mathbb{U}}\underaccent{\bar}{\mathbb{X}}}}\\
			\textbf{Rule 3:}& \probac{\mathbb{v}^{\mathbb{y}}}{\interv{\mathbb{v}^{\mathbb{u}}},\interv{\mathbb{v}^{\mathbb{x}}},\mathbb{v}^{\mathbb{w}}} = \probac{\mathbb{v}^{\mathbb{y}}}{\interv{\mathbb{v}^{\mathbb{u}}},\mathbb{v}^{\mathbb{w}}}\\
			&\text{if } \dsepc{\mathbb{Y}}{\mathbb{X}}{\mathbb{U}, \mathbb{W}}{\mathcal{G}^s_{\bar{\mathbb{U}}\bar{\mathbb{X}(\mathbb{W})}}}\\
		\end{aligned}
	\end{equation*}
	where $\mathbb{X}(\mathbb{W})$ is the set of vertices in $\mathbb{X}$ that are non-ancestors of any vertex in $\mathbb{W}$ in the mutilated graph $\mathcal{G}^s_{\bar{\mathbb{U}}}$.
\end{theorem}

\begin{figure}[t!]
	\centering
	\begin{subfigure}{.15\textwidth}
		\centering
		\begin{tikzpicture}[{black, circle, draw, inner sep=0}]
			\tikzset{nodes={draw,rounded corners},minimum height=0.6cm,minimum width=0.6cm}	
			\tikzset{anomalous/.append style={fill=easyorange}}
			\tikzset{rc/.append style={fill=easyorange}}
			
			\node[fill=red!30] (X) at (0,0) {$X$} ;
			\node[fill=blue!30] (Y) at (1.2,0) {$Y$};
			
			\begin{scope}[transform canvas={yshift=-.25em}]
				\draw [->,>=latex] (X) -- (Y);
			\end{scope}
			\begin{scope}[transform canvas={yshift=.25em}]
				\draw [<-,>=latex] (X) -- (Y);
			\end{scope}
			
			\draw[->,>=latex] (X) to [out=180,in=135, looseness=2] (X);

            \draw[<->,>=latex, dashed] (X) to [out=-155,in=-110, looseness=2] (X);
        \draw[<->,>=latex, dashed] (Y) to [out=-25,in=-70, looseness=2] (Y);

		\end{tikzpicture}
		\caption{\centering}
		\label{fig:non_identifiable:a}
	\end{subfigure}
	\hfill 
	\begin{subfigure}{.15\textwidth}
		\centering
		\begin{tikzpicture}[{black, circle, draw, inner sep=0}]
			\tikzset{nodes={draw,rounded corners},minimum height=0.6cm,minimum width=0.6cm}	
			\tikzset{anomalous/.append style={fill=easyorange}}
			\tikzset{rc/.append style={fill=easyorange}}
			
			\node[fill=red!30] (X) at (0,0) {$X$} ;
			\node[fill=blue!30] (Y) at (1.2,0) {$Y$};
			
			\draw [->,>=latex] (X) -- (Y);
			\draw[<->,>=latex, dashed] (X) to [out=90,in=90, looseness=1] (Y);
			
			\draw[->,>=latex] (X) to [out=180,in=135, looseness=2] (X);
			\draw[->,>=latex] (Y) to [out=15,in=60, looseness=2] (Y);

    			\draw[<->,>=latex, dashed] (X) to [out=-155,in=-110, looseness=2] (X);
			\draw[<->,>=latex, dashed] (Y) to [out=-25,in=-70, looseness=2] (Y);

		\end{tikzpicture}
		\caption{\centering}
		\label{fig:non_identifiable:b}
	\end{subfigure}
	\hfill 
	\begin{subfigure}{.15\textwidth}
		\centering
		\begin{tikzpicture}[{black, circle, draw, inner sep=0}]
			\tikzset{nodes={draw,rounded corners},minimum height=0.6cm,minimum width=0.6cm}	
			\tikzset{anomalous/.append style={fill=easyorange}}
			\tikzset{rc/.append style={fill=easyorange}}
			
			\node (W) at (0,1.2) {$W$} ;
			\node[fill=red!30] (X) at (0,0) {$X$} ;
			\node[fill=blue!30] (Y) at (1.2,0) {$Y$};
			
			\draw [->,>=latex] (X) -- (Y);
			\begin{scope}[transform canvas={xshift=-.25em}]
				\draw [->,>=latex] (X) -- (W);
			\end{scope}
			\begin{scope}[transform canvas={xshift=.25em}]
				\draw [<-,>=latex] (X) -- (W);
			\end{scope}
			\draw[<->,>=latex, dashed] (W) to [out=0,in=90, looseness=1] (Y);
			
			\draw[->,>=latex] (W) to [out=180,in=135, looseness=2] (W);
			\draw[->,>=latex] (X) to [out=180,in=135, looseness=2] (X);
			\draw[->,>=latex] (Y) to [out=15,in=60, looseness=2] (Y);

  			\draw[<->,>=latex, dashed] (X) to [out=-25,in=-70, looseness=2] (X);
			\draw[<->,>=latex, dashed] (Y) to [out=-25,in=-70, looseness=2] (Y);
			\draw[<->,>=latex, dashed] (W) to [out=20,in=65, looseness=2] (W);			

		\end{tikzpicture}
		\caption{\centering}
		\label{fig:non_identifiable:c}
	\end{subfigure}
	% \hfill 
	% 	\begin{subfigure}{.28\textwidth}
		% 		\centering
		% 		\begin{tikzpicture}[{black, circle, draw, inner sep=0}]
			% 			\tikzset{nodes={draw,rounded corners},minimum height=0.6cm,minimum width=0.6cm}	
			% 			\tikzset{latent/.append style={white, fill=black}}
			
			% 			% \node[latent] (L) at (1.75,1.2) {$L$};
			% 			\node[fill=red!30] (X) at (0,0) {$X$} ;
			% 			\node[fill=blue!30] (Y) at (3.5,0) {$Y$};
			% 			\node (Z) at (1.75,0) {$W$};

			%  \begin{scope}[transform canvas={yshift=-.25em}]
				%  \draw [->,>=latex] (X) -- (Z);
				%  \end{scope}
			%  \begin{scope}[transform canvas={yshift=.25em}]
				%  \draw [<-,>=latex] (X) -- (Z);
				%  \end{scope}			
			
			%  \begin{scope}[transform canvas={yshift=-.25em}]
				%  \draw [->,>=latex] (Z) -- (Y);
				%  \end{scope}
			%  \begin{scope}[transform canvas={yshift=.25em}]
				%  \draw [<-,>=latex] (Z) -- (Y);
				%  \end{scope}			

			%             \draw[<->,>=latex, dashed] (X) to [out=90,in=90, looseness=1] (Y);
			
			% 			\draw[->,>=latex] (X) to [out=165,in=120, looseness=2] (X);
			% 			\draw[->,>=latex] (Y) to [out=15,in=60, looseness=2] (Y);
			% 			% \draw[->,>=latex] (L) to [out=165,in=120, looseness=2] (L);
			% 			\draw[->,>=latex] (Z) to [out=15,in=60, looseness=2] (Z);
			
			% 		\end{tikzpicture}
		% 		\caption{}
		% 		\label{fig:non_identifiable:d}
		% 	\end{subfigure}
	\caption{SCGs with not identifiable macro total effects. Each pair of red and blue vertices represents the total effect we are interested in.}
	\label{fig:non_identifiable}
\end{figure}
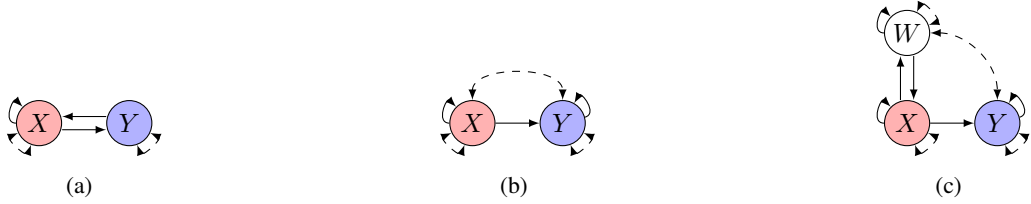%

\begin{proof}
	Let $\mathcal{G}^s=(\mathbb{S},\mathbb{E}^s)$ be an SCG and $\mathbb{X},\mathbb{Y},\mathbb{U},\mathbb{W}\subseteq\mathbb{S}$.
	
	Suppose that $\dsepc{\mathbb{Y}}{\mathbb{X}}{\mathbb{U}, \mathbb{W}}{\mathcal{G}^s_{\bar{\mathbb{U}}}}$.
	Using Property~\ref{prop:mutilated_graphs} and Theorem \ref{theorem:soundness_d-sep} we know that for every compatible FT-ADMG $\mathcal{G}$, $\dsepc{\mathbb{V}^{\mathbb{Y}}}{\mathbb{V}^{\mathbb{X}}}{\mathbb{V}^{\mathbb{U}}, \mathbb{V}^{\mathbb{W}}}{\mathcal{G}_{\bar{\mathbb{V}^{\mathbb{U}}}}}$.
	Thus, the usual rule 1 of the do-calculus given by \citealt{Pearl_1995} in the case of ADMGs states that $\probac{\mathbb{v}^{\mathbb{y}}}{\interv{\mathbb{v}^{\mathbb{u}}},\mathbb{v}^{\mathbb{x}},\mathbb{v}^{\mathbb{w}}} = \probac{\mathbb{v}^{\mathbb{y}}}{\interv{\mathbb{v}^{\mathbb{u}}},\mathbb{v}^{\mathbb{w}}}$.
	
	Suppose that $\dsepc{\mathbb{Y}}{\mathbb{X}}{\mathbb{U}, \mathbb{W}}{\mathcal{G}^s_{\bar{\mathbb{U}}\underaccent{\bar}{\mathbb{X}}}}$.
	Using Property~\ref{prop:mutilated_graphs} and Theorem \ref{theorem:soundness_d-sep} we know that for every compatible FT-ADMG $\mathcal{G}$, $\dsepc{\mathbb{V}^{\mathbb{Y}}}{\mathbb{V}^{\mathbb{X}}}{ \mathbb{V}^{\mathbb{U}},\mathbb{V}^{\mathbb{W}}}{\mathcal{G}_{\bar{\mathbb{V}^{\mathbb{U}}}\underaccent{\bar}{\mathbb{V}^\mathbb{X}}}}$.
	Thus, the usual rule 2 of the do-calculus given by \citealt{Pearl_1995} in the case of ADMGs states that $\probac{\mathbb{v}^{\mathbb{y}}}{\interv{\mathbb{v}^{\mathbb{u}}},\interv{\mathbb{v}^{\mathbb{x}}},\mathbb{v}^{\mathbb{w}}} = \probac{\mathbb{v}^{\mathbb{y}}}{\interv{\mathbb{v}^{\mathbb{u}}},\mathbb{v}^{\mathbb{x}},\mathbb{v}^{\mathbb{w}}}$.
	
	Suppose that $\dsepc{\mathbb{Y}}{\mathbb{X}}{\mathbb{U}, \mathbb{W}}{\mathcal{G}^s_{\bar{\mathbb{U}}\underaccent{\bar}{\mathbb{X}(\mathbb{W}})}}$.
	Using Property~\ref{prop:mutilated_graphs}, $\mathcal{G}^s_{\bar{\mathbb{U}}}$ is compatible with any $\mathcal{G}_{\bar{\mathbb{V}^{\mathbb{U}}}}$ and if $\exists X_{t_X}\in\mathbb{V}^{\mathbb{X}}$ such that $X_{t_X}\in \ancestors{\mathbb{V}^{\mathbb{W}}}{\mathcal{G}_{\bar{\mathbb{V}^{\mathbb{U}}}}}$ then $X\in \ancestors{\mathbb{W}}{\mathcal{G}^s_{\bar{\mathbb{U}}}}$. Therefore, $\forall X\in\mathbb{X},~ X\notin \ancestors{\mathbb{W}}{\mathcal{G}^s_{\bar{\mathbb{U}}}} \implies \forall X_{t_X} \in \mathbb{V}^{\mathbb{X}},~ X_{t_X}\notin \ancestors{\mathbb{V}^{\mathbb{X}}}{\mathcal{G}_{\bar{\mathbb{V}^{\mathbb{U}}}}}$.
	Thus, $\mathcal{G}^s_{\bar{\mathbb{U}}\bar{\mathbb{X}(\mathbb{W}})}$ and $\mathcal{G}_{\bar{\mathbb{V}^{\mathbb{U}}}\bar{\mathbb{V}^{\mathbb{X}}(\mathbb{V}^{\mathbb{W}}})}$ are compatible, and using Theorem \ref{theorem:soundness_d-sep} we know that for every compatible FT-ADMG $\mathcal{G}$, $\dsepc{\mathbb{V}^{\mathbb{Y}}}{\mathbb{V}^{\mathbb{X}}}{\mathbb{V}^{\mathbb{U}},\mathbb{V}^{\mathbb{W}}}{\mathcal{G}_{\bar{\mathbb{V}^{\mathbb{U}}}\bar{\mathbb{V}^{\mathbb{X}}(\mathbb{V}^{\mathbb{W}}})}}$.
	Thus, the usual rule 3 of the do-calculus given by \citealt{Pearl_1995} in the case of ADMGs states that $\probac{\mathbb{v}^{\mathbb{y}}}{\interv{\mathbb{v}^{\mathbb{u}}},\interv{\mathbb{v}^{\mathbb{x}}},\mathbb{v}^{\mathbb{w}}} = \probac{\mathbb{v}^{\mathbb{y}}}{\interv{\mathbb{v}^{\mathbb{u}}},\mathbb{v}^{\mathbb{w}}}$.
\end{proof}

Using the soundness of the do-calculus in SCGs (Theorem~\ref{theorem:soundness_do-calculus}), we can easily use the rules of the do-calculus to find out that the total effect $\probac{\mathbb{v}^\mathbb{y}}{\interv{\mathbb{v}^\mathbb{x}}}$
is identifiable in all SCGs in Figure~\ref{fig:identifiable}.
Both in Figure~\ref{fig:identifiable:a} and~\ref{fig:identifiable:c}, one can verify that $\dsep{\mathbb{Y}}{\mathbb{X}}{\mathcal{G}^s_{\underaccent{\bar}{{\mathbb{X}}}}}$, thus Rule 2 of the do-calculus is applicable and $\probac{\mathbb{v}^\mathbb{y}}{\interv{\mathbb{v}^\mathbb{x}}} = \probac{\mathbb{v}^\mathbb{y}}{\mathbb{v}^\mathbb{x}}$. Notice that Figure~\ref{fig:identifiable:b} does not contain any cycle other than self-loops and is very similar to Figure 1(b) of \cite{Anand_2023} which corresponds to the well-known front-door criterion~\citep{Pearl_2000}. Thus, using the corresponding sequence of classical rules of probability and rules of do-calculus as the one given in \citep[p.83]{Pearl_2000}, one obtains $\probac{\mathbb{v}^\mathbb{y}}{\interv{\mathbb{v}^\mathbb{x}}} = \sum_{\mathbb{v}^\mathbb{w}} \probac{\mathbb{v}^\mathbb{w}}{\mathbb{v}^\mathbb{x}} \sum_{\mathbb{v}^\mathbb{x'}} \probac{\mathbb{v}^\mathbb{y}}{\mathbb{v}^\mathbb{w},\mathbb{v}^\mathbb{x'}} \proba{\mathbb{v}^\mathbb{x'}}$. In these three examples, the rules of do-calculus allow macro interventional distributions to be expressed solely in terms of observational distributions. Consequently, the macro total effect can be estimated from the data, provided the positivity assumption holds.

In the following, we introduce the second main theorem of this section.

\begin{theorem}[Completeness of the do-calculus]
	\label{theorem:completeness_do-calculus}
	If one of the do-calculus rules does not apply for a given SCG, then there exists a compatible FT-ADMG for which the corresponding rule does not apply.
\end{theorem}

\begin{proof}
	Let $\mathcal{G}^s=(\mathbb{S},\mathbb{E}^s)$ be an SCG and $\mathbb{X},\mathbb{Y},\mathbb{U},\mathbb{W}\subseteq\mathbb{S}$.
	
	Suppose that rule 1 does not apply \ie, $\notdsepc{\mathbb{Y}}{\mathbb{X}}{\mathbb{U}, \mathbb{W}}{\mathcal{G}^s_{\bar{\mathbb{U}}}}$.
	Then, using Theorem~\ref{theorem:completeness_d-sep} there exists an FT-ADMG $\tilde{\mathcal{G}}$ compatible with $\mathcal{G}^s_{\bar{\mathbb{U}}}$ in which $\notdsepc{\mathbb{V}^{\mathbb{Y}}}{\mathbb{V}^{\mathbb{X}}}{\mathbb{V}^{\mathbb{U}}, \mathbb{V}^{\mathbb{W}}}{\tilde{\mathcal{G}}}$.
	Notice that there exists an FT-ADMG $\mathcal{G}$ compatible with $\mathcal{G}^s$ such that $\mathcal{G}_{\bar{\mathbb{V}^{\mathbb{U}}}} = \tilde{\mathcal{G}}$.
	Therefore, $\mathcal{G}$ is an FT-ADMG compatible with $\mathcal{G}^s$ in which $\notdsepc{\mathbb{V}^{\mathbb{Y}}}{\mathbb{V}^{\mathbb{X}}}{\mathbb{V}^{\mathbb{U}}, \mathbb{V}^{\mathbb{W}}}{\mathcal{G}_{\bar{\mathbb{V}^{\mathbb{U}}}}}$ and thus 
 %in which 
 the usual rule 1 of the do-calculus given by \citealt{Pearl_1995} in the case of ADMGs does not apply.
	
	Suppose that rule 2 does not apply \ie, $\notdsepc{\mathbb{Y}}{\mathbb{X}}{\mathbb{U}, \mathbb{W}}{\mathcal{G}^s_{\bar{\mathbb{U}}\underaccent{\bar}{\mathbb{X}}}}$.
	Then, using Theorem~\ref{theorem:completeness_d-sep} there exists an FT-ADMG $\tilde{\mathcal{G}}$ compatible with $\mathcal{G}^s_{\bar{\mathbb{U}}\underaccent{\bar}{\mathbb{X}}}$ in which $\notdsepc{\mathbb{V}^{\mathbb{Y}}}{\mathbb{V}^{\mathbb{X}}}{\mathbb{V}^{\mathbb{U}}, \mathbb{V}^{\mathbb{W}}}{\tilde{\mathcal{G}}}$.
	Notice that there exists an FT-ADMG $\mathcal{G}$ compatible with $\mathcal{G}^s$ such that $\mathcal{G}_{\bar{\mathbb{V}^{\mathbb{U}}}\underaccent{\bar}{\mathbb{V}^{\mathbb{X}}}} = \tilde{\mathcal{G}}$.
	Therefore, $\mathcal{G}$ is an FT-ADMG compatible with $\mathcal{G}^s$ in which $\notdsepc{\mathbb{V}^{\mathbb{Y}}}{\mathbb{V}^{\mathbb{X}}}{\mathbb{V}^{\mathbb{U}}, \mathbb{V}^{\mathbb{W}}}{\mathcal{G}_{\bar{\mathbb{V}^{\mathbb{U}}}\underaccent{\bar}{\mathbb{V}^{\mathbb{X}}}}}$ and thus 
 %in which 
 the usual rule 2 of the do-calculus given by \citealt{Pearl_1995} in the case of ADMGs does not apply.
	
	Suppose that rule 3 does not apply \ie, $\notdsepc{\mathbb{Y}}{\mathbb{X}}{\mathbb{U}, \mathbb{W}}{\mathcal{G}^s_{\bar{\mathbb{U}}\bar{\mathbb{X}(\mathbb{W})}}}$.
	Then, using Theorem~\ref{theorem:completeness_d-sep} there exists an FT-ADMG $\tilde{\mathcal{G}}$ compatible with $\mathcal{G}^s_{\bar{\mathbb{U}}\bar{\mathbb{X}(\mathbb{W})}}$ in which $\notdsepc{\mathbb{V}^{\mathbb{Y}}}{\mathbb{V}^{\mathbb{X}}}{\mathbb{V}^{\mathbb{U}}, \mathbb{V}^{\mathbb{W}}}{\tilde{\mathcal{G}}}$.
	Using the same idea as in the proof of Theorem~\ref{theorem:soundness_do-calculus}, notice that there exists an FT-ADMG $\mathcal{G}$ compatible with $\mathcal{G}^s$ such that $\mathcal{G}_{\bar{\mathbb{V}^{\mathbb{U}}}\bar{\mathbb{V}^{\mathbb{X}}(\mathbb{V}^{\mathbb{W}})}} = \tilde{\mathcal{G}}$.
	Therefore, $\mathcal{G}$ is an FT-ADMG compatible with $\mathcal{G}^s$ in which $\notdsepc{\mathbb{V}^{\mathbb{Y}}}{\mathbb{V}^{\mathbb{X}}}{\mathbb{V}^{\mathbb{U}}, \mathbb{V}^{\mathbb{W}}}{\mathcal{G}_{\bar{\mathbb{V}^{\mathbb{U}}}\bar{\mathbb{V}^{\mathbb{X}}(\mathbb{V}^{\mathbb{W}})}}}$ and thus 
 %in which 
 the usual rule 3 of the do-calculus given by \citealt{Pearl_1995} in the case of ADMGs does not apply.
\end{proof}

Using the completeness of the do-calculus in SCGs (Theorem~\ref{theorem:completeness_do-calculus}), we can determine that the total effect $\probac{\mathbb{v}^\mathbb{y}}{\interv{\mathbb{v}^\mathbb{x}}}$  is not identifiable in all SCGs depicted in Figure~\ref{fig:non_identifiable} and in the SCG in Figure~\ref{fig:SCG_0} by examining all possible iterations of the rules of the do-calculus.
However, it is well known that exhaustively examining all possibilities for applying the rules of do-calculus can quickly become impractical, particularly for large graphs. To address this challenge, the following subsection introduces a sub-graphical structure designed to directly determine whether it is feasible to express an interventional distribution solely in terms of observational distributions using do-calculus and SCGs.

%Note that even if the do-calculus is sound and complete for identifying macro total effects in SCGs, the ID algorithm~\citep{Shpitser_2008} which is based on the do-calculus and complete for directed acyclic graphs cannot be applied in the case of SCGs even for macro total effects.  The ID algorithm rely on the notion of Hedges \citep{Shpitser_2008} to know if the total effect is identifiable or not. However, this notion of Hedges is not sufficient for cyclic graphs. For example, the SCG in Figure~\ref{fig:non_identifiable:a} contains no Hedge but the macro total effect is not identifiable due to the cycle between $X$ and $Y$.

% \begin{proposition}
	%     The ID algorithm is not correct for identifying macro total effect in SCGs.
	% \end{proposition}

\subsection{Strongly Connected Hedges and Non-identifiability}
In ADMGs, there exists a sub-graphical structure, called an Hedge~\citep{Shpitser_2006}, which is employed to  graphically characterize non-identifiability as shown in \citealp[Theorem 4]{Shpitser_2006}. 
To properly define it for SCGs, it is essential to first familiarize oneself with the two related definitions which we have adapted and provided below specifically for the context of SCGs:

\begin{definition}[C-component,  \citealt{Tian_2002}]
    Let $\mathcal{G}^s=(\mathbb{S},\mathbb{E}^s)$ be an SCG. A subset of vertices $\mathbb{V} \subseteq \mathbb{S}$ such that $\forall V,V' \in \mathbb{V},~\exists V^1,\cdots,V^n\in \mathbb{V}$ with $V^1=V$, $V^n=V'$ and $\forall i\in[1,n-1]~V^i\longdashleftrightarrow V^{i+1}$ is called a C-component.
\end{definition}
\begin{definition}[C-forest, \citealt{Shpitser_2006}]
    Let $\mathcal{G}^s=(\mathbb{S},\mathbb{E}^s)$ be an SCG. A subgraph\footnote{Following \citealt{Shpitser_2006}, we consider that a subgraph of $\mathcal{G}^s$ is a graph containing a subset of the vertices of $\mathcal{G}^s$ and a subset of the edges between those vertices in $\mathcal{G}^s$.} of  $\mathcal{G}^s=(\mathbb{S},\mathbb{E}^s)$ which is acyclic, where every vertex has at most one child (\ie, is a forest), and is a C-component is called a C-forest.
\end{definition}

\begin{definition}[Hedge,  \citealt{Shpitser_2006,Shpitser_2008}]
Consider an SCG $\mathcal{G}^s=(\mathbb{S}, \mathbb{E}^s)$ and two sets of vertices $\mathbb{X}, \mathbb{Y}\subset \mathbb{S}$.
Let $\mathbb{F}$ and $\mathbb{F}'$ be R-rooted C-forests in $\mathcal{G}^s$ such that $\mathbb{X} \cap \mathbb{F} \ne\emptyset$, $\mathbb{X} \cap \mathbb{F}' =\emptyset$, $\mathbb{F}' \subseteq \mathbb{F}$, and $R\subset \ancestors{Y}{\mathcal{G}^s_{\underaccent{\bar}{\mathbb{X}}}}$. Then $\mathbb{F}$ and $\mathbb{F}'$ form an Hedge for $\probac{\mathbb{v}^\mathbb{y}}{\interv{\mathbb{v}^\mathbb{x}}}$ in $\mathcal{G}^s$.
\end{definition}

%The concept of a 
An Hedge turned out to be too weak to cover non-identifiability in SCGs. %(and other cyclic graphs). 
For example, the SCG in Figure~\ref{fig:non_identifiable:a} contains no Hedge but the macro total effect is not identifiable due to the cycle between $X$ and $Y$.
Consequently, we introduce a new notion designed to graphically characterize non-identifiability of macro total effects in SCGs.

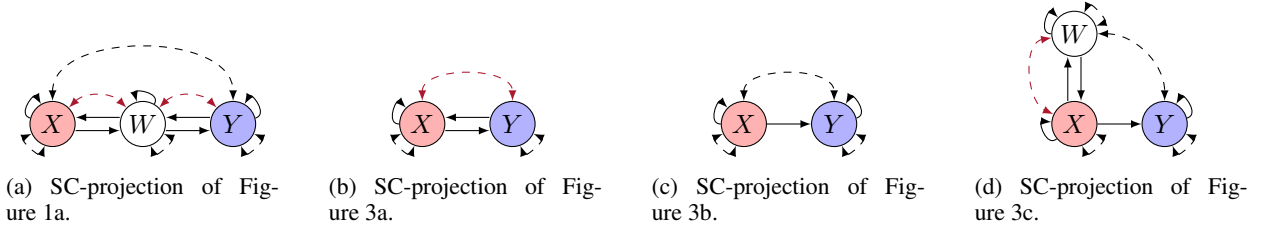
\begin{figure}[t!]
	\centering
 	\begin{subfigure}{.22\textwidth}
		\centering
		\begin{tikzpicture}[{black, circle, draw, inner sep=0}]
			\tikzset{nodes={draw,rounded corners},minimum height=0.6cm,minimum width=0.6cm}	
			\tikzset{latent/.append style={white, fill=black}}
			
			% \node[latent] (L) at (1.75,1.2) {$L$};
			\node[fill=red!30] (X) at (0,0) {$X$} ;
			\node[fill=blue!30] (Y) at (2.4,0) {$Y$};
			\node (Z) at (1.2,0) {$W$};

			\begin{scope}[transform canvas={yshift=-.25em}]
				\draw [->,>=latex] (X) -- (Z);
			\end{scope}
			\begin{scope}[transform canvas={yshift=.25em}]
				\draw [<-,>=latex] (X) -- (Z);
			\end{scope}			
			
			\begin{scope}[transform canvas={yshift=-.25em}]
				\draw [->,>=latex] (Z) -- (Y);
			\end{scope}
			\begin{scope}[transform canvas={yshift=.25em}]
				\draw [<-,>=latex] (Z) -- (Y);
			\end{scope}

			% \draw[->,>=latex, dashed] (L) to [out=0,in=90, looseness=1] (Y);
			% \draw[->,>=latex, dashed] (L) to [out=180,in=90, looseness=1] (X);
			\draw[<->,>=latex, dashed] (X) to [out=90,in=90, looseness=1] (Y);
			
			\draw[->,>=latex] (X) to [out=165,in=120, looseness=2] (X);
			\draw[->,>=latex] (Y) to [out=15,in=60, looseness=2] (Y);
			% \draw[->,>=latex] (L) to [out=165,in=120, looseness=2] (L);
			\draw[->,>=latex] (Z) to [out=60,in=105, looseness=2] (Z);
			
			\draw[<->,>=latex, dashed] (X) to [out=-155,in=-110, looseness=2] (X);
			\draw[<->,>=latex, dashed] (Y) to [out=-25,in=-70, looseness=2] (Y);
			\draw[<->,>=latex, dashed] (Z) to [out=-25,in=-70, looseness=2] (Z);			
   			\draw[<->,>=latex, dashed, sorbonnered] (X) to [out=45,in=135, looseness=1] (Z);
   			\draw[<->,>=latex, dashed, sorbonnered] (Z) to [out=45,in=135, looseness=1] (Y);

		\end{tikzpicture}
		\caption{SC-projection of  Figure~\ref{fig:SCG_0}.}
		\label{fig:proj_SCG_0}
	\end{subfigure}
        \hfill 
	\begin{subfigure}{.22\textwidth}
		\centering
		\begin{tikzpicture}[{black, circle, draw, inner sep=0}]
			\tikzset{nodes={draw,rounded corners},minimum height=0.6cm,minimum width=0.6cm}	
			\tikzset{anomalous/.append style={fill=easyorange}}
			\tikzset{rc/.append style={fill=easyorange}}
			
			\node[fill=red!30] (X) at (0,0) {$X$} ;
			\node[fill=blue!30] (Y) at (1.2,0) {$Y$};
			
			\begin{scope}[transform canvas={yshift=-.25em}]
				\draw [->,>=latex] (X) -- (Y);
			\end{scope}
			\begin{scope}[transform canvas={yshift=.25em}]
				\draw [<-,>=latex] (X) -- (Y);
			\end{scope}
			
			\draw[->,>=latex] (X) to [out=180,in=135, looseness=2] (X);

			\draw[<->,>=latex, dashed] (X) to [out=-155,in=-110, looseness=2] (X);
			\draw[<->,>=latex, dashed] (Y) to [out=-25,in=-70, looseness=2] (Y);

   		\draw[<->,>=latex, dashed, sorbonnered] (X) to [out=90,in=90, looseness=1] (Y);

		\end{tikzpicture}
		\caption{SC-projection of  Figure~\ref{fig:non_identifiable:a}.}
		\label{fig:proj_non_identifiable:a}
	\end{subfigure}

 \begin{subfigure}{.22\textwidth}
		\centering
		\begin{tikzpicture}[{black, circle, draw, inner sep=0}]
			\tikzset{nodes={draw,rounded corners},minimum height=0.6cm,minimum width=0.6cm}	
			\tikzset{anomalous/.append style={fill=easyorange}}
			\tikzset{rc/.append style={fill=easyorange}}
			
			\node[fill=red!30] (X) at (0,0) {$X$} ;
			\node[fill=blue!30] (Y) at (1.2,0) {$Y$};
			
			\draw [->,>=latex] (X) -- (Y);
			\draw[<->,>=latex, dashed] (X) to [out=90,in=90, looseness=1] (Y);
			
			\draw[->,>=latex] (X) to [out=180,in=135, looseness=2] (X);
			\draw[->,>=latex] (Y) to [out=15,in=60, looseness=2] (Y);

  			\draw[<->,>=latex, dashed] (X) to [out=-155,in=-110, looseness=2] (X);
			\draw[<->,>=latex, dashed] (Y) to [out=-25,in=-70, looseness=2] (Y);

		\end{tikzpicture}
		\caption{SC-projection of  Figure~\ref{fig:non_identifiable:b}.}
		\label{fig:proj_non_identifiable:b}
	\end{subfigure}
	\hfill 
	\begin{subfigure}{.22\textwidth}
		\centering
		\begin{tikzpicture}[{black, circle, draw, inner sep=0}]
			\tikzset{nodes={draw,rounded corners},minimum height=0.6cm,minimum width=0.6cm}	
			\tikzset{anomalous/.append style={fill=easyorange}}
			\tikzset{rc/.append style={fill=easyorange}}
			
			\node (W) at (0,1.2) {$W$} ;
			\node[fill=red!30] (X) at (0,0) {$X$} ;
			\node[fill=blue!30] (Y) at (1.2,0) {$Y$};
			
			\draw [->,>=latex] (X) -- (Y);
			\begin{scope}[transform canvas={xshift=-.25em}]
				\draw [->,>=latex] (X) -- (W);
			\end{scope}
			\begin{scope}[transform canvas={xshift=.25em}]
				\draw [<-,>=latex] (X) -- (W);
			\end{scope}
			\draw[<->,>=latex, dashed] (W) to [out=0,in=90, looseness=1] (Y);
			
			\draw[->,>=latex] (W) to [out=180,in=135, looseness=2] (W);
			\draw[->,>=latex] (X) to [out=220,in=175, looseness=2] (X);
			\draw[->,>=latex] (Y) to [out=15,in=60, looseness=2] (Y);

			\draw[<->,>=latex, dashed] (X) to [out=-25,in=-70, looseness=2] (X);
			\draw[<->,>=latex, dashed] (Y) to [out=-25,in=-70, looseness=2] (Y);
			\draw[<->,>=latex, dashed] (W) to [out=20,in=65, looseness=2] (W);			

   			\draw[<->,>=latex, dashed, sorbonnered] (W) to [out=190,in=155, looseness=1] (X);

		\end{tikzpicture}
		\caption{SC-projection of  Figure~\ref{fig:non_identifiable:c}.}
		\label{fig:proj_non_identifiable:c}
	\end{subfigure}
	\caption{SC-projections of the SCGs in Figures~\ref{fig:SCG_FTADMG} and \ref{fig:non_identifiable}. Each pair of red and blue vertices represents the total effect we are interested in, and the red edges indicate those added through the SC-projection.}
	\label{fig:projections}
\end{figure}%

\begin{definition}[Strongly connected projection (SC-projection)]
\label{def:SC-proj}
    Consider an SCG $\mathcal{G}^s=\{\mathbb{S}, \mathbb{E}^s\}$.
    The SC-projection $\mathcal{H}^s$ of $\mathcal{G}^s$ is the graph that includes all vertices and edges from $\mathcal{G}^s$, plus a bidirected dashed edge between each pair $X,Y\in \mathbb{S}$ such that $\scc{X}{\mathcal{G}^s} = \scc{Y}{\mathcal{G}^s}$ and $X\ne Y$.
\end{definition}

\begin{definition}[Strongly connected Hedge (SC-Hedge)]
\label{def:SC-Hedge}
%    Consider an SCG $\mathcal{G}^s=\{\mathbb{S}, \mathbb{E}^s\}$ and two sets of vertices $\mathbb{X}, \mathbb{Y}\subset \mathbb{S}$. Let $\mathcal{H}^s$ be the graph that includes all vertices and edges from $\mathcal{G}^s$, plus a bidirected dashed edge between each pair $X,Y\in \mathbb{S}$ such that $\scc{X}{\mathcal{G}^s} = \scc{Y}{\mathcal{G}^s}$. We say that, an Hedge in $\mathcal{H}^s$ is an SC-Hedge in $\mathcal{G}^s$.
    Consider an SCG $\mathcal{G}^s=\{\mathbb{S}, \mathbb{E}^s\}$ and its SC-projection $\mathcal{H}^s$.
    An Hedge for $\probac{\mathbb{v}^\mathbb{y}}{\interv{\mathbb{v}^\mathbb{x}}}$ in $\mathcal{H}^s$ is an SC-Hedge for $\probac{\mathbb{v}^\mathbb{y}}{\interv{\mathbb{v}^\mathbb{x}}}$ in $\mathcal{G}^s$.
\end{definition}
\begin{theorem}
\label{theorem:SC-Hedge}
    If there exists a SC-Hedge  for $\probac{\mathbb{v}^\mathbb{y}}{\interv{\mathbb{v}^\mathbb{x}}}$ in $\mathcal{G}^s$ then $\probac{\mathbb{v}^\mathbb{y}}{\interv{\mathbb{v}^\mathbb{x}}}$ is not identifiable.
\end{theorem}
\begin{proof}
    Consider an SCG $\mathcal{G}^s=(\mathbb{S},\mathbb{E})$, its SC-projection $\mathcal{H}^s$, and a SC-Hedge $\mathbb{F},\mathbb{F}'$ for $\probac{\mathbb{v}^\mathbb{y}}{\interv{\mathbb{v}^\mathbb{x}}}$.
    Let us prove Theorem~\ref{theorem:SC-Hedge} by induction on the number of bidirected dashed edges in the C-forest $\mathbb{F}$ that are in $\mathcal{H}^s$ but not in $\mathcal{G}^s$ (\ie, which are artificially induced by cycles).
    Firstly, if $\mathbb{F},\mathbb{F}'$ is an Hedge in $\mathcal{G}^s$ then there exists a compatible FT-ADMG in which for any $\exists t\in [t_0,t_{max}],~\mathbb{F}_{t} = \{F_{t} \mid F \in \mathbb{F}\},\mathbb{F}'_{t} = \{F_{t} \mid F \in \mathbb{F}'\}$ is an Hedge for $\probac{\mathbb{v}^\mathbb{y}}{\interv{\mathbb{v}^\mathbb{x}}}$ and thus $\probac{\mathbb{v}^\mathbb{y}}{\interv{\mathbb{v}^\mathbb{x}}}$ is not identifiable.
    Secondly, if there exists $X\in\mathbb{X}$ and $Y\in\mathbb{Y}$ such that $\scc{X}{\mathcal{G}^s} = \scc{Y}{\mathcal{G}^s}$ then there exists a compatible FT-ADMG in which $\exists t\in [t_0,t_{max}]$ there is a directed path from $Y_{t}$ to $X_{t}$. This path must be blocked but every vertex on this path is a descendant of $Y_{t}$ and thus cannot be adjusted on without inducing a bias.
    Lastly, assume that Theorem~\ref{theorem:SC-Hedge} is true for any SC-Hedge with $k$ bidirected dashed edges which are in $\mathcal{H}^s$ but not in $\mathcal{G}^s$ and that $\mathbb{F}$ has $k+1$ such edges. Then, one cannot identify the macro total effect by the adjustment formula~\cite{Shpitser_2010} due to the ambiguity \citep{Assaad_2024} induced by the cycle on $X$ or due to the bias induced by the latent confounder of $X$ that cannot be removed without using any rule of the do-calculus. Moreover, any decomposition of the effect using the do-calculus will necessitate the identification of other macro total effects. These effects have sub-C-forests of $\mathbb{F},\mathbb{F}'$ as SC-Hedges and at least one of theses sub-C-forests has at most $k$ artificial bidirected dashed edges induced by cycles. The associated total macro effect is therefore unidentifiable by induction.
\end{proof}

For example, each SCG depicted in Figures~\ref{fig:SCG_FTADMG} and \ref{fig:non_identifiable} features an SC-Hedge. This is apparent from their SC-projections shown in Figure~\ref{fig:projections}, where an Hedge is clearly present. This implies, based on Theorem~\ref{theorem:SC-Hedge}, that in these cases, the total effect is not identifiable.

%(in \ref{fig:non_identifiable:a} and \ref{fig:non_identifiable:b} $F$ contains $X$ and $Y$ and $F'$ contains $Y$ and in \ref{fig:non_identifiable:c} $F$ contains $X$,$Y$ and $W$ and $F'$ contains $Y$).

%\section{What about micro queries?}
%\label{sec:micro}

\section{Related Works}
\label{sec:discussion}
The idea of representing complex low-level causal structures as partially specified high-level graphs has gained some interest in the last years. 
For instance,
%Closely related to our work,
\citealt{iwasaki_1994,chalupka_2016,Anand_2023,Wahl_2024,Ferreira_2024,Assaad_2024} have explored clustering of low-level variables. Notably, \citealt{Anand_2023} introduced cluster-ADMGs, which generalize acyclic directed mixed graphs by allowing flexible partitioning of variables while maintaining acyclicity. They extended d-separation and the do-calculus to accommodate macro causal effects in these graphs. Even if SCGs, discussed in our paper, are somehow related (in both graphs vertices represent clusters) to cluster-ADMGs they are inherently different since SCGs can contain cycles.
%%%
%It should be noted that \citealt{Richardson_1997, Forre_2017, Forre_2018, Forre_2020} extended d-separation and the do-calculus to DMGs. Although not explicitly focusing on cluster-DMGs, their findings might apply to macro queries in such graphs. Indeed, SCGs can be considered as a special case of cluster-DMGs, however, the applicability of the results in \citealt{Richardson_1997, Forre_2017, Forre_2018, Forre_2020} to SCGs is not that clear due to SCGs' additional temporal information (there is a known causal order within each cluster), which may affect completeness for macro queries but not necessarily soundness. Our research indicates that  this temporal knowledge does not enhance the identifiability of macro queries.
%%%
It should be noted that \citealt{Richardson_1997, Forre_2017, Forre_2018, Forre_2020} extended d-separation and do-calculus to DMGs, potentially applicable to macro queries in cluster-DMGs. Although SCGs are a type of cluster-DMG, the applicability of these findings to SCGs remains uncertain due to the specific temporal information within each cluster in an SCGs, which might affect the completeness for macro queries but not necessarily their soundness. Our research suggests that this temporal knowledge does not improve the identifiability of macro queries.

In dynamic systems, \citealt{Assaad_2023,Ferreira_2024,Assaad_2024} considered SCGs but for micro queries. They showed that the total effect and the direct effect is identifiable under some mild conditions if there are no hidden confounding which is different from macro causal effects as they are always identifiable when there is no hidden confounding. 
As far as we know there exists still no sound and complete result for d-separation and the do-calculus for SCG in the more general case for micro queries, \ie,
the results presented in this paper do not apply to micro queries. For example, in the context of conditional independencies, even if $\dsepc{X}{Y}{\mathbb{W}}{\mathcal{G}^s}$ in the SCG, we cannot infer that $\dsepc{X_{t-1}}{Y_t}{\tilde{\mathbb{V}}^{\mathbb{W}}}{\mathcal{G}}$ where $\mathcal{G}$ is compatible with $\mathcal{G}^s$ and $\tilde{\mathbb{V}}^{\mathbb{W}}$ does not include all variables represented by $\mathbb{V}^{\mathbb{W}}$. 
In the context of total effects, \citealt{Assaad_2024} demonstrated that the total effect $\probac{y_t}{\interv{x_{t-1}}}$ is identifiable even if there is a bidirectional arrow $X \rightleftarrows Y$ in the SCG, provided there are no other cycles involving $X$ and $Y$ as in Figure~\ref{fig:non_identifiable:a}.
However, the do-calculus rules described in this paper cannot identify micro queries. It is not complete in the sense that it does not identify every identifiable queries.
In the context of macro queries within dynamic systems, \citealt{Reiter_2024} explored the identifiability of causal effects from one spectrum (or frequency) to another using SCGs, without delving into the graphical conditions like the ones presented in this paper. Furthermore, in parallel (and independently) to this work, \citealt{Boeken_2024} examined a more general scenario where time is treated as continuous, and instantaneous cycles are possible in the full-time graph.

%Furthermore, there exists also other related works that include \citealt{scholkopf_2021, Shen_2018, parviainen_2016} but we consider them as beyond the scope of the paper.

In a different view, \citealt{Rubenstein_2017, Beckers_2019} studied more general notions of abstraction, namely exact ($\tau-\omega$)-transformations and $\tau$-abstractions.
If one is only focused in the macro effects then the SCG can be seen as a constructive $\tau$-abstraction where $\tau$ is the identity in $\mathbb{R}^{|[t_0,t_{max}]|}$ (assuming the micro variables take values in $\mathbb{R}$).
Which means that all contributions of this paper can contribute to the development of the theory of more general notion of abstraction.
However, this is no longer true if one is interested in the micro effects.

\section{Conclusion}
\label{sec:conclusion}

In this paper, %we have presented an analysis of summary causal graphs (SCGs) and their properties.
% We demonstrated how SCGs are different from other known partially specified graph, in particular $\tau$-abstraction and cluster directed acyclic graphs.
we established the soundness and completeness of d-separation and the do-calculus for respectively identifying macro conditional independencies and macro total effects in SCGs. 
By doing so, we bridged the gap between SCGs and many real-world applications in epidemiology. 
%Overall, the confirmation of soundness and completeness of macro d-separation and macro total effects in SCGs enhances the utility of causal graphs in dynamic systems, offering a powerful tool for finding conditional independencies and inferring total effects with greater confidence.
%We also showed through various examples that the results of this paper do not hold when considering micro conditional independencies and micro total effects.

There are three main limitations to this work. The first limitation is that the completeness result in Theorem~\ref{theorem:completeness_do-calculus} does not take into account that there might exists different iterations of the rules of the do-calculus in different FT-ADMGs that can give the same final identification of the total effect. 
A second related limitation is that we provided a graphical characterization for the non-identifiability of macro total effects, however this characterization is not necessarily complete, therefore in future works we want to see if a complete characterization might involve SC-Hedges in some way. 
The third limitation is that our results require a complete observation in time of the entire event system, which is often not feasible in practice. For example, it is challenging to ensure that data collection begins at the onset of each epidemic outbreak. Without this assumption, past instances of observed variables might act as unobserved confounders.
% One can overcome this limitation by assuming the absence or negligibility of such confounders.
In these situations, it would be worthwhile to explore the potential use of instrumental variables~\citep{Pearl_2000} to address these challenges. When instrumental variables are unavailable, our results may still be useful for \emph{partially} identifying macro total effects using SCGs by deriving bounds for the effect under a minimal set of assumptions~\citep{Robins_1989, Manski_1990, Pearl_2000, Zhang_2021}.

% \textcolor{red}{
% Furthermore, it would be important to investigate whether the ID Algorithm~\citep{Shpitser_2006,Shpitser_2008,Forre_2020} can be extended to automatically identify a macro total effect in SCGs whenever it is identifiable and returns failts whenever it is not.
% }

\appendix

\section{Acknowledgments}
We thank Fabrice Carrat from IPLESP for his valuable discussions regarding the application of this work in epidemiology. Additionally, we thank Philip Boeken and Joris Mooij from KdVI for their valuable input on dynamic structural causal models and cyclic graphs. Last but not least, we are thankful to Jonas Wahl from TU Berlin for his insightful discussions regarding cluster graphs and summary causal graphs in the frequency domain. This work was supported by the CIPHOD project (ANR-23-CPJ1-0212-01).

\bibliography{references.bib}

\newpage

\section{Appendix}

As in \citealt{Ferreira_2024}, we define the notion of primary path of a walk $\tilde{\pi}=\langle V^1,\cdots,V^n \rangle$ as $\pi = \langle U^1,\cdots,U^m \rangle$ where $U^1 = V^1$ and $U^{k+1} = V^{max\{i\mid V^i = U^k\}+1}$.

\propertyone*

\begin{proof}
	Let $\mathcal{G}=(\mathbb{V},\mathbb{E})$ be an FT-ADMG, $\mathbb{W}\subseteq\mathbb{V}$ and $\tilde{\pi}=\langle V^1,\cdots,V^n\rangle$ be an $\mathbb{W}$-active walk.
	Suppose that its primary path $\pi = \langle U^1,\cdots,U^m \rangle$ is $\mathbb{W}$-blocked and take $1<i<m$ such that $\langle U^{i-1},U^i,U^{i+1}\rangle$ is $\mathbb{W}$-blocked.

	If $U^{i-1}\starbarstar U^{i}\rightarrow U^{i+1}$ (or symmetrically if $U^{i-1} \leftarrow U^{i} \starbarstar U^{i+1}$) and $U^{i}\in\mathbb{W}$ then there exists $1<j<n$ such that $U^i=V^j$ and $U^{j+1}=V^{j+1}$ and thus $V^{j-1}\starbarstar V^{j}\rightarrow V^{j+1}$ is in $\tilde{\pi}$ and $V^j=U^{i}\in\mathbb{W}$.
	Therefore $\tilde{\pi}$ is $\mathbb{W}$-blocked.

	Otherwise, $U^{i-1}\stararrow U^{i}\arrowstar U^{i+1}$ and $\descendants{U^{i}}{\mathcal{G}}\cap\mathbb{W}=\emptyset$.
	Then, there exists $1<j<n$ such that $U^{i-1}=V^{j-1}$ and $U^{i}=V^{j}$ and thus $ V^{j-1} \stararrow V^j \cdots  \arrowstar $ is in $\tilde{\pi}$.
	Take $j\leq k < n$ such that $V^{j-1} \stararrow V^j \rightarrow \cdots \rightarrow V^{k} \arrowstar$ is in $\tilde{\pi}$ and notice that $\stararrow V^k \arrowstar$ and $\descendants{V^k}{\mathcal{G}} \subseteq \descendants{V^{j}}{\mathcal{G}} = \descendants{U^{i}}{\mathcal{G}}$ so $\descendants{V^k}{\mathcal{G}} \cap \mathbb{W} = \emptyset$.
	Therefore, $\tilde{\pi}$ is $\mathbb{W}$-blocked which contradicts the initial assumption.
\end{proof}

\propertytwo*

\begin{proof}
	Let $\mathcal{G}=(\mathbb{V},\mathbb{E})$ be an FT-ADMG, $\mathcal{G}^s=(\mathbb{S},\mathbb{E}^s)$ its compatible SCG and $\mathbb{A},\mathbb{B}\subseteq\mathbb{S}$.
	
	Firstly, let us show that every arrow in $\mathcal{G}^ s_{\bar{\mathbb{A}}\underaccent{\bar}{\mathbb{B}}}$ corresponds to an arrow in $\mathcal{G}_{\bar{\mathbb{V}^{\mathbb{A}}}\underaccent{\bar}{\mathbb{V}^{\mathbb{B}}}}$.
	Let $X,Y\in\mathbb{S}$ such that $X\rightarrow Y$ is in $\mathbb{E}^s_{\bar{\mathbb{A}}\underaccent{\bar}{\mathbb{B}}}$.
	Because $X\rightarrow Y$ is in $\mathbb{E}^s_{\bar{\mathbb{A}}\underaccent{\bar}{\mathbb{B}}}$, we know that $X\notin\mathbb{B}$ and $Y\notin \mathbb{A}$ and that $X\rightarrow Y$ is in $\mathbb{E}^s$.
	Therefore, there exists $t_0\leq t_X \leq t_Y \leq t_{max}$ such that $X_{t_X}\rightarrow Y_{t_Y}$ is in $\mathbb{E}$ and $X_{t_X}\notin\mathbb{V}^{\mathbb{B}}$ and $Y_{t_Y}\notin \mathbb{V}^{\mathbb{A}}$ so $X_{t_X}\rightarrow Y_{t_Y}$ is in $\mathbb{E}_{\bar{\mathbb{V}^{\mathbb{A}}}\underaccent{\bar}{\mathbb{V}^{\mathbb{B}}}}$.
	
	Similarly, let $X,Y\in\mathbb{S}$ such that $X\longdashleftrightarrow Y$ is in $\mathbb{E}^s_{\bar{\mathbb{A}}\underaccent{\bar}{\mathbb{B}}}$.
	Because $X\longdashleftrightarrow Y$ is in $\mathbb{E}^s_{\bar{\mathbb{A}}\underaccent{\bar}{\mathbb{B}}}$, we know that $X,Y\notin \mathbb{A}$ and that $X\longdashleftrightarrow Y$ is in $\mathbb{E}^s$.
	Therefore, there exists $t_X , t_Y \in [t_0,t_{max}]$ such that $X_{t_X}\longdashleftrightarrow Y_{t_Y}$ is in $\mathbb{E}$ and $X_{t_X},Y_{t_Y}\notin \mathbb{V}^{\mathbb{A}}$ so $X_{t_X}\longdashleftrightarrow Y_{t_Y}$ is in $\mathbb{E}_{\bar{\mathbb{V}^{\mathbb{A}}}\underaccent{\bar}{\mathbb{V}^{\mathbb{B}}}}$.
	
	Secondly, let us show that every arrow in $\mathcal{G}_{\bar{\mathbb{V}^{\mathbb{A}}}\underaccent{\bar}{\mathbb{V}^{\mathbb{B}}}}$ corresponds to an arrow in $\mathcal{G}^ s_{\bar{\mathbb{A}}\underaccent{\bar}{\mathbb{B}}}$.
	Let $X_{t_X},Y_{t_Y}\in\mathbb{V}_{\bar{\mathbb{V}^{\mathbb{A}}}\underaccent{\bar}{\mathbb{V}^{\mathbb{B}}}}$ such that $X_{t_X}\rightarrow Y_{t_Y}$ is in $\mathbb{E}_{\bar{\mathbb{V}^{\mathbb{A}}}\underaccent{\bar}{\mathbb{V}^{\mathbb{B}}}}$.
	Because $X_{t_X}\rightarrow Y_{t_Y}$ is in $\mathbb{E}_{\bar{\mathbb{V}^{\mathbb{A}}}\underaccent{\bar}{\mathbb{V}^{\mathbb{B}}}}$, we know that $X_{t_X}\notin\mathbb{V}^{\mathbb{B}}$ and $Y_{t_Y}\notin \mathbb{V}^{\mathbb{A}}$ and that $X_{t_X}\rightarrow Y_{t_Y}$ is in $\mathbb{E}$.
	Therefore, $X\rightarrow Y$ is in $\mathbb{E}^s$ and $X\notin{\mathbb{B}}$ and $Y\notin {\mathbb{A}}$ so $X\rightarrow Y$ is in $\mathbb{E}^s_{\bar{\mathbb{A}}\underaccent{\bar}{\mathbb{B}}}$.
	
	Similarly, let $X_{t_X},Y_{t_Y}\in\mathbb{V}_{\bar{\mathbb{V}^{\mathbb{A}}}\underaccent{\bar}{\mathbb{V}^{\mathbb{B}}}}$ such that $X_{t_X}\longdashleftrightarrow Y_{t_Y}$ is in $\mathbb{E}_{\bar{\mathbb{V}^{\mathbb{A}}}\underaccent{\bar}{\mathbb{V}^{\mathbb{B}}}}$.
	Because $X_{t_X}\longdashleftrightarrow Y_{t_Y}$ is in $\mathbb{E}_{\bar{\mathbb{V}^{\mathbb{A}}}\underaccent{\bar}{\mathbb{V}^{\mathbb{B}}}}$, we know that $X_{t_X},Y_{t_Y}\notin \mathbb{V}^{\mathbb{A}}$ and that $X_{t_X}\longdashleftrightarrow Y_{t_Y}$ is in $\mathbb{E}$.
	Therefore, $X\longdashleftrightarrow Y$ is in $\mathbb{E}^s$ and $X,Y\notin {\mathbb{A}}$ so $X\longdashleftrightarrow Y$ is in $\mathbb{E}^s_{\bar{\mathbb{A}}\underaccent{\bar}{\mathbb{B}}}$.
\end{proof}

\end{document}